\documentclass[11pt]{amsart}
\usepackage{amsmath}
\usepackage[dvips]{epsfig}
\usepackage{slashed}
\usepackage{accents}
\usepackage{color}
\usepackage{leftidx}

\theoremstyle{plain}
\newtheorem{theorem}{Theorem}[section]
\newtheorem{lemma}[theorem]{Lemma}

\theoremstyle{definition}
\newtheorem{remark}[theorem]{Remark}

\numberwithin{equation}{section}
\newcommand{\volg}{1+ u^{2}| \nabla f |_{g}^{2}} %volume form of \bg%
\newcommand{\volbarg}{1-u^{2}| \overline{\nabla} f |_{\overline{g}}^{2}}  %volume form of  g%
\newcommand{\barna}{\overline{\nabla}} %barred nabla - covariant derivative w.r.t. barred g%
\newcommand{\bg}{\overline{g}}  %barred g%
\newcommand{\by}{\overline{Y}}  %barred Y%
\newcommand{\byp}{Y^{\phi}}  % Y^\phi%
\newcommand{\bk}{\overline{k}} %barred k - second fundamental form for t=0%
\newcommand{\be}{\overline{E}}
\newcommand{\bm}{\overline{B}}

\begin{document}

\title[Geometric Inequalities in Asymptotically Hyperboloidal Slices] {Reduction Arguments for Geometric Inequalities Associated With Asymptotically Hyperboloidal Slices}

\author[Cha]{Ye Sle Cha}
\address{Institut f\"ur Mathematik \\
 Freie Universit\"at Berlin \\
 14195 Berlin, Germany}
\email{ycha@zedat.fu-berlin.de}

\author[Khuri]{Marcus Khuri}
\address{Department of Mathematics\\
Stony Brook University\\
Stony Brook, NY 11794, USA}
\email{khuri@math.sunysb.edu}

\author[Sakovich]{Anna Sakovich}
%\address{Mathematical Sciences Research Institute\\Berkeley, CA 94720, USA}
\address{Max-Planck-Institut f\"ur Gravitationsphysik (Albert-Einstein-Institute)\\
 Am M\"uhlenberg 1, 14476 Potsdam, Germany}
\email{sakovich.ann@gmail.com}

\thanks{M. Khuri acknowledges the support of
NSF Grant DMS-1308753.}

\begin{abstract}
We consider several geometric inequalities in general relativity involving mass, area, charge, and angular momentum for asymptotically hyperboloidal initial data. We show how to reduce each one to the known maximal (or time symmetric) case in the asymptotically flat setting, whenever a geometrically motivated system of elliptic equations admits a solution.
\end{abstract}
\maketitle

\section{Introduction}
\label{sec1} \setcounter{equation}{0}
\setcounter{section}{1}

In \cite{SchoenYau}, Schoen and Yau proved the spacetime version of the positive mass theorem for asymptotically flat initial data by utilizing a reduction procedure involving the so called Jang equation \cite{Jang}. For asymptotically hyperboloidal slices of asymptotically flat spacetimes, a similar reduction argument was given \cite{SchoenYau1} (see also \cite{HuangYauZhang}) in which solutions of the Jang equation are required to possess hyperboloidal asymptotics. This type of solution to the Jang equation results in a deformation of the initial data, which transforms the original asymptotically hyperbolic structure into an asymptotically flat structure and preserves the mass up to multiplication by a positive constant. In addition, the Jang equation imparts a positivity property to the scalar curvature of the deformed data. As in \cite{SchoenYau} this yields a conformal change of metric to zero scalar curvature, from which one may conclude nonnegativity of the mass as a consequence of the time symmetric case \cite{SchoenYau0} of the positive mass theorem in the asymptotically flat setting. The existence of solutions to the Jang equation with the desired hyperboloidal asymptotics has recently been established in \cite{Sakovich} (see also \cite{SakovichThesis}).

In this paper we seek to generalize this strategy of transforming asymptotically hyperboloidal data to asymptotically flat data, so that it may be applied to several geometric inequalities motivated by the standard picture of gravitational collapse and the (weak) cosmic censorship conjecture \cite{Choquet-Bruhat,Penrose}. Namely, the inequalities that shall be treated here include the Penrose inequality \cite{Bray,HuiskenIlmanen,Mars}, the Penrose inequality with charge
\cite{Jang1,KhuriWeinsteinYamada,KhuriWeinsteinYamada0,KhuriWeinsteinYamada1},
the positive mass theorem with charge
\cite{BartnikChrusciel,ChruscielReallTod,GHHP}, the mass-angular momentum inequality \cite{ChruscielLiWeinstein,Dain0,SchoenZhou}, the mass-angular momentum-charge inequality \cite{ChruscielCosta,Costa,KhuriWeinstein1,SchoenZhou}, as well as a lower bound for the area of black holes in terms of mass, angular momentum, and charge \cite{DainKhuriWeinsteinYamada}.
In the asymptotically flat setting, the general case of each of these inequalities may be reduced to the known maximal (or time symmetric) case by solving a canonical system of equations specific to each inequality \cite{BrayKhuri1,BrayKhuri2,ChaKhuri1,ChaKhuri2,DisconziKhuri,KhuriWeinstein}. The main equation involved shares a resemblance to the classical Jang equation, and the solution is chosen to vanish at spatial infinity.
In the setting of asymptotically hyperbolic data arising from asymptotically hyperboloidal slices, we will show that these Jang-type equations may be solved with the hyperboloidal asymptotics produced in \cite{SakovichThesis, Sakovich} for the classical Jang equation. Thus, with a similar procedure, all these inequalities for asymptotically hyperboloidal slices are reduced to solving a canonical system of equations.

\section{Notation and Definitions}
\label{sec2} \setcounter{equation}{0}
\setcounter{section}{2}

Consider an initial data set $(M,g,k)$ for the Einstein equations modeling an asymptotically hyperboloidal, spacelike hypersurface, in an asymptotically flat spacetime. This consists of a Riemannian 3-manifold $M$ with asymptotically hyperbolic metric $g$, and asymptotically umbilic extrinsic curvature $k$. We define such an initial data set to be asymptotically hyperboloidal if it possesses an end which is
diffeomorphic to $S^{2}\times [r_{0},\infty)$, and in this region there are coordinates such that
\begin{equation}\label{1}
g=g_{0}+a, \text{ }\text{ }\text{ }\text{ }\text{ }k=g_{0} + b,
\end{equation}
where $g_{0}=\frac{dr^{2}}{1+r^{2}}+r^{2}\sigma$ is the hyperbolic metric with $\sigma$ the round metric on $S^2$, and
\begin{equation}\label{2}
a_{rr}=\frac{\mathbf{m}^{r}}{r^5}+O_{3}(r^{-6}), \text{ }\text{ }\text{ }\text{ }\text{ }\text{ } a_{r\alpha}=O_{3}(r^{-3}), \text{ }\text{ }\text{ }\text{ }\text{ }\text{ } a_{\alpha\beta}=\frac{\mathbf{m}^{g}_{\alpha\beta}}{r}+ O_{3}(r^{-2}),
\end{equation}
%\begin{equation}\label{2.1}
%\widetilde{a}_{\alpha\beta},\widetilde{b}_{\alpha\beta}=O(r^{-2}),\text{ }\text{ }\text{ %}\text{ }\text{ }\partial_{\gamma}\widetilde{a}_{\alpha\beta}=O(r^{-2}),\text{ }\text{ %}\text{ }\text{ }\text{ }\partial_{r}\widetilde{a}_{\alpha\beta}=O(r^{-3}),
%\end{equation}
\begin{equation}\label{3}
b_{rr}=O_{2}(r^{-5}), \text{ }\text{ }\text{ }\text{ }\text{ }\text{ }  b_{r\alpha}=O_{2}(r^{-3}),\text{ }\text{ }\text{ }\text{ }\text{ }\text{ } b_{\alpha\beta}= \frac{\mathbf{m}^{k}_{\alpha \beta}}{r} + O_{2}(r^{-2}),
\end{equation}
with Greek letters denoting indices for coordinates on $S^{2}$.
Here $\mathbf{m}^{g}$, $\mathbf{m}^{k}$ are tensors and $\mathbf{m}^{r}$ is a function, all on $S^{2}$ and independent of $r$. The notation $h=O_{l}(r^{-n})$ asserts that $r^{n+i}|\partial^{i}_{r}\partial^{\gamma}h|\leq C$ for all $i+|\gamma|\leq l$, and also for later use
$h=o_{l}(r^{-n})$ asserts that
$\lim_{r\rightarrow\infty}r^{n+i}\partial_{r}^{i}\partial^{\gamma}h=0$
for all $i+|\gamma|\leq l$. We note that these assumptions can be reformulated, and even weakened in the context of Sections \ref{sec3}-\ref{sec5} below, by using weighted H\"older spaces, see e.g. \cite{DahlSakovich}.

The quantities $\mathbf{m}^{g}$, $\mathbf{m}^{k}$ and $\mathbf{m}^{r}$ encode mass through the formula
\begin{equation}\label{3.1}
m=\frac{1}{16\pi}\int_{S^{2}}\left[Tr_{\sigma}\left( \mathbf{m}^{g} + 2\mathbf{m}^{k} \right)+2\mathbf{m}^{r}\right],
\end{equation}
where the integrand is the so called mass aspect function. Note that in referring to $m$ as mass this is a slight abuse of terminology, as this quantity is physically the total energy, that is, the first component of the energy-momentum vector (see \cite[Definition 1.4]{ChenWangYau}). Definitions of the energy-momentum vector in the case of asymptotically hyperboloidal initial data with more general asymptotic behavior at infinity can be found in \cite{ChruscielJezierskiLeski}, \cite{Michel}.

\begin{remark}
Clearly, the quantity $m$ defined by \eqref{3.1} allows interpretation as mass, that is length of energy-momentum vector, provided that the coordinate chart at infinity is such that the linear momentum vanishes. Such a coordinate chart at infinity is called balanced. Note that if the energy-momentum vector is timelike with respect to a given chart $\Phi$ at infinity then there is an isometry $I$ of the hyperbolic space $(\mathbb{H}^3, g_0)$ such that the chart $I \circ \Phi$ is balanced. This is a consequence of the fact that the energy-momentum vector transforms equivariantly with respect to Lorentz boosts which in turn restrict to (nonlinear) isometries of the hyperbolic space.
\end{remark}

The initial data also satisfy the constraint equations
\begin{equation}\label{4}
	2\mu = R+(Tr_{g} k)^{2}-|k|_{g}^{2},
	\text{ }\text{ }\text{ }\text{ }\text{ }\text{ } J = \operatorname{div}_{g}(k- (Tr_{g} k) g),
\end{equation}
where $\mu$ and $J$ are the energy and momentum density of the matter fields, and $R$ is scalar curvature. Moreover, the dominant energy condition is given by
\begin{equation}\label{5}
\mu \geq |J|_{g}.
\end{equation}

For geometric inequalities involving electromagnetic charge, we will make use of initial data for the Einstein-Maxwell equations $(M,g,k,E,B)$. Here $E$ and $B$ are vector fields representing the induced electric and magnetic field on the slice. Such data will also be referred to as asymptotically hyperboloidal, if in addition to the requirements above the electromagnetic field satisfies
\begin{equation}\label{6}
  E_{r},B_{r} = O(r^{-3}),\text{ }\text{ }\text{ }\text{ }\text{ }E_{\alpha},B_{\alpha}=O(r^{-1})\text{ }\text{ }\text{ }\text{ }\text{ }
  \Rightarrow\text{ }\text{ }\text{ }\text{ }|E|_{g}+|B|_{g}=O(r^{-2}).
\end{equation}
The energy and momentum density of the non-electromagnetic matter fields is given by
\begin{equation}\label{7}
  	2\mu_{EM} = R+(Tr_{g} k)^{2}-|k|_{g}^{2} - 2(|E|_{g}^2+|B|_{g}^{2}),\text{ }\text{ }\text{ }\text{ }\text{ }\text{ } J_{EM} = \operatorname{div}_{g}(k- (Tr_{g} k) g)+2E\times B,
\end{equation}
where $(E\times B)_{i}=\epsilon_{ijl}E^{j}B^{l}$ is the cross product with
$\epsilon$ the volume form of $g$. The quantities $\operatorname{div}_{g}E$ and $\operatorname{div}_{g}B$ are interpreted as the electric and magnetic charge density, and the following inequality will be referred to as the charged dominant energy condition
\begin{equation}\label{10}
   \mu_{EM} \geq |J_{EM}|_{g} + \frac{1}{2}\left(|\operatorname{div}_{g}E|
   +|\operatorname{div}_{g}B|\right).
\end{equation}
Note that
\begin{equation}\label{8}
\mathcal{Q}_{e} = \frac1{4\pi} \int_{S_\infty} g(E,\nu_{g})\, , \qquad
\mathcal{Q}_{b} = \frac1{4\pi} \int_{S_\infty} g(B,\nu_{g})\, ,
\end{equation}
are well-defined in light of the fall-off conditions \eqref{6},
where $S_{\infty}$ indicates the limit as $r\rightarrow\infty$ of integrals over coordinate spheres $S_{r}$, with unit
outer normal $\nu_{g}$. Here $\mathcal{Q}_{e}$ and $\mathcal{Q}_{b}$ denote the total electric and magnetic charge respectively,
and we denote the square of the total charge by $\mathcal{Q}^{2}=\mathcal{Q}_{e}^{2}+\mathcal{Q}_{b}^{2}$.

When the initial data have a boundary $\partial M$, it will consist of an outermost apparent horizon. That is, each boundary component $S\subset\partial M$ satisfies $\theta_{+}(S):=H_{S}+Tr_{S}k=0$ (future horizon) or $\theta_{-}(S):=H_{S}-Tr_{S}k=0$ (past horizon), where $H$ denotes mean curvature with respect to the normal pointing towards null infinity, and
no other apparent horizons are present. Furthermore, in some cases the initial data will have ends which are not asymptotically hyperboloidal, but are rather asymptotically flat or asymptotically cylindrical. We say that an end is asymptotically flat if it is diffeomorphic to $\mathbb{R}^{3}\setminus\mathrm{Ball}$, and in the Cartesian coordinates $x^i$ given by this diffeomorphism the following fall-off conditions hold
\begin{equation}\label{9}
|g_{ij}-\delta_{ij}|+|x||\partial g_{ij}| + |x|^2|\partial\partial g_{ij}|=O(|x|^{-1}),\text{ }\text{ }\text{
}\text{ }
|k_{ij}|+ |E_{i}|+|B_{i}|=O(|x|^{-2})\text{ }\text{ }\text{as}\text{ }\text{
}|x|\rightarrow\infty.
\end{equation}
%\begin{equation}\label{9}
%|(g_{ij}-\delta_{ij})|+|x||\partial g_{ij}|=O(|x|^{-1}),\text{ }\text{ }\text{
%}\text{ }|k_{ij}|=O(|x|^{-2}),\text{ }\text{ }\text{
%}\text{ }|E_{i}|+|B_{i}|=O(|x|^{-2})\text{ }\text{ }\text{as}\text{ }\text{
%}|x|\rightarrow\infty.
%\end{equation}
For such an end, the ADM mass is well-defined and given by
\begin{equation}\label{10.1}
m_{adm}=\frac{1}{16\pi}\int_{S_{\infty}}(\partial_{i}g_{ij}-\partial_{j}g_{ii})\nu^{j},
\end{equation}
where again we abuse terminology as this is physically the energy.
The asymptotics of cylindrical ends is most conveniently described in Brill coordinates, see Section \ref{sec6}.

In each section that follows, we will detail a reduction argument for a different geometric inequality associated with asymptotically hyperboloidal initial data. Hence, each inequality will be reduced to solving a canonical system of equations. Moreover, we will show that the primary (or Jang-type) equation in each system may be solved independently with the desired asymptotics needed for the procedure.

\section{The Penrose Inequality}
\label{sec3} \setcounter{equation}{0}
\setcounter{section}{3}

Well known heuristic arguments of Penrose \cite{Penrose, Penrose1} lead to the Penrose inequality for asymptotically null slices in asymptotically flat
spacetimes. Given asymptotically hyperboloidal initial data $(M,g,k)$ satisfying the dominant energy condition,
the Penrose inequality \cite{Mars} states that
\begin{equation}\label{14}
m\geq\sqrt{\frac{A}{16\pi}}
\end{equation}
where $A$ is the minimum area required to enclose the outermost apparent horizon. We will denote the region outside of the outermost minimal area enclosure of the horizon by $\widetilde{M}$, so that $A=|\partial\widetilde{M}|$.

Consider a graph $\overline{M}=\{t=f(x)\}$ inside the warped product 4-manifold $(M \times \mathbb{R}, g + u^2 dt^2)$, then the induced metric on $\overline{M}$ is given by $\overline{g}=g+u^{2}df^{2}$. Here $u$ is a nonnegative function to be chosen appropriately.
If the generalized Jang equation
\begin{equation}\label{15}
\left(g^{ij}-\frac{u^{2}f^{i}f^{j}}{1+u^{2}|\nabla f|_{g}^{2}}\right)
\left(\frac{u\nabla_{ij}f+u_{i}f_{j}+u_{j}f_{i}}
{\sqrt{1+u^{2}|\nabla f|_{g}^{2}}}-k_{ij}\right)=0
\end{equation}
is satisfied, then $\overline{M}$ is referred to as the Jang surface and the Jang metric $\overline{g}$ obtains a desirable positivity property for its scalar curvature.
%where $\nabla$ denotes covariant differentiation with respect to the metric %$g$, $f_{i}=\partial_{i}f$,
%and $f^{i}=g^{ij}f_{j}$.  This equation is quasi-linear elliptic, and %degenerates when either $\phi=0$ or $f$ blows-up.
In particular, the scalar curvature of the Jang graph \cite{BrayKhuri1,BrayKhuri2} is weakly nonnegative and given by
\begin{equation}\label{15.1}
\overline{R}=2(\mu-J(w))+
|\pi-k|_{\overline{g}}^{2}+2|q|_{\overline{g}}^{2}
-2u^{-1}\operatorname{div}_{\overline{g}}(u q),
\end{equation}
where
$\pi$ is the second fundamental form of $\overline{M}$ in the dual Lorentzian setting $(M \times \mathbb{R}, \overline{g}-u^2 dt^2)$,
and $w$ and $q$ are 1-forms given by
\begin{equation}\label{15.2}
\pi_{ij}=\frac{ u \nabla_{ij}f
+ u_i f_j +  u_j  f_i}{ \sqrt{1 + u^2 |\nabla f|_g^2 }},\text{ }\text{ }\text{ }\text{ }
w_{i}=\frac{u f_{i}}{\sqrt{1+u^{2}|\nabla f|_{g}^{2}}},\text{
}\text{ }\text{ }\text{ }
q_{i}=\frac{u f^{j}}{\sqrt{1+u^{2}|\nabla f|_{g}^{2}}}(\pi_{ij}-k_{ij}).
\end{equation}
The generalized Jang equation was introduced to study the (non-time-symmetric) Penrose inequality in the asymptotically flat case. In this setting $f$ vanishes at spatial infinity, and boundary conditions are imposed on $\partial\widetilde{M}$ to guarantee that the boundary of the Jang surface $\partial\Sigma$ is minimal; these boundary conditions often entail a blow-up of the Jang graph and are described in \cite{BrayKhuri2}, \cite{HanKhuri2}.
The existence, regularity, and blow-up behavior for the generalized Jang equation is studied at length in \cite{HanKhuri1}.
Also in this setting, the warping function $u$
is assumed to vanish on $\partial\widetilde{M}$, and to have an expansion in the asymptotic end of the form
\begin{equation}\label{15.3}
u= 1 + \frac{u_{0}}{r}+O_{2}\left(\frac{1}{r^{2-\varepsilon}}\right),
\end{equation}
where $u_{0}$ is a constant.

In the asymptotically hyperboloidal setting addressed here, the boundary conditions on $\partial\widetilde{M}$ and the expansion \eqref{15.3} will remain unchanged, however the behavior of the Jang graph in the asymptotic end will be completely different. Namely, in analogy with the approach to the positive mass theorem \cite{SchoenYau1}, \cite{Sakovich} we impose the following asymptotics at null infinity
\begin{equation}\label{11}
f(r,\theta,\phi)=\sqrt{1+r^{2}}+\mathcal{A}\log r+\mathcal{B}(\theta,\phi)+\widetilde{f}(r,\theta,\phi),
\end{equation}
where $(\theta,\phi)$ are coordinates on $S^{2}$,
\begin{equation}\label{12}
\mathcal{A}=2m,\text{ }\text{ }\text{ }\text{ }\text{ }\text{ }
\Delta_{\sigma}\mathcal{B}=\frac{1}{2}\left[Tr_{\sigma}( \mathbf{m}^g + 2\mathbf{m}^k )+2\mathbf{m}^r\right]-\frac{1}{8\pi}\int_{S^{2}}\left[Tr_{\sigma}( \mathbf{m}^g + 2\mathbf{m}^k )+2\mathbf{m}^r\right],
\end{equation}
and
\begin{equation}\label{13}
\widetilde{f} = O_{3}(r^{-1+\varepsilon})
%\partial_{\alpha}\widetilde{f}=O(r^{-\varepsilon}),\text{ }\text{ }
%\partial_{r}\widetilde{f}=O(r^{-1-\varepsilon}),\text{ }\text{ }
%\partial_{\alpha}\partial_{\beta}\widetilde{f}=O(r^{-\varepsilon}),\text{ }\text{ }
%\partial_{r}\partial_{\alpha}\widetilde{f}=O(r^{-1-\varepsilon}),\text{ }\text{ }
%\partial_{r}^{2}\widetilde{f}=O(r^{-2-\varepsilon}).
\end{equation}
for any $\varepsilon >0$.

\begin{lemma}\label{lemma1}
If $(M,g,k)$ is asymptotically hyperboloidal and \eqref{15.3}-\eqref{13} are satisfied then the Jang metric $\overline{g}=g+u^{2}df^{2}$ is asymptotically flat, and the mass of the Jang metric is given by $\overline{m}_{adm}=2m+u_{0}$.
\end{lemma}

\begin{proof}
It is clear that the manifold $(\overline{M},\overline{g})$ has an end diffeomorphic to $R^{3} \setminus Ball$, with coordinates $y=(r,\theta,\phi)$ as in \eqref{2}, \eqref{3}. Let $x$ denote the associated Cartesian coordinates,  related to $y$ through the usual (spherical coordinates) transformation. In what follows, $i,j$ are indices for $x$-coordinates and $a,b$ are indices for $y$-coordinates. It follows that
\begin{equation}\label{3.16}
\overline{g}_{ij}=g_{ij}+u^{2}f_{i}f_{j}
=\left(g_{ab}+u^{2}f_{a}f_{b}\right)\frac{\partial y^{a}}
{\partial x^{i}}\frac{\partial y^{b}}
{\partial x^{j}}.
\end{equation}
From \eqref{2}, \eqref{3}, \eqref{15.3}, \eqref{11}, and \eqref{13} we have
\begin{equation}\label{3.17}
g_{rr}+u^{2}f_{r}^{2}=1+O(r^{-1}),\text{ }\text{ }\text{ }\text{ }\text{ }
g_{r\alpha}+u^{2}f_{r}f_{\alpha}=O(1),\text{ }\text{ }\text{ }\text{ }\text{ }
g_{\alpha\beta}+u^{2}f_{\alpha}f_{\beta}
=r^{2}\sigma_{\alpha\beta}+O(1),
\end{equation}
which implies that
\begin{equation}\label{3.18}
\left(g_{ab}+u^{2}f_{a}f_{b}\right)\frac{\partial y^{a}}
{\partial x^{i}}\frac{\partial y^{b}}
{\partial x^{j}}=\delta_{ij}+O(r^{-1})
\end{equation}
since $\frac{\partial r}{\partial x^{i}}=O(1)$, and $\frac{\partial y^{\alpha}}{\partial x^{i}}=O(r^{-1})$. Estimates on the derivatives of \eqref{3.16} may be obtained in a similar fashion, from which \eqref{9} follows.

The mass $\overline{m}_{adm}$ of the Jang metric may be computed as in \cite{Sakovich}. Note that we can write \eqref{10.1} in a coordinate free form as
\begin{equation}
\overline{m}_{adm}=\frac{1}{16\pi}\int_{S_{\infty}}(\operatorname{div}_\delta \overline{g} - d Tr_{\delta} \overline{g}) (\nu),
\end{equation}
where $\delta$ is the Euclidean metric.  It is convenient to compute this integral in the spherical coordinates $y$. In this case we have
\begin{equation}
\accentset{\circ}{\Gamma}^{r}_{rr} = \accentset{\circ}{\Gamma}^\alpha_{rr} = \accentset{\circ}{\Gamma}^r_{\alpha r} = 0, \qquad \accentset{\circ}{\Gamma}^r_{\alpha\beta} = -r \sigma_{\alpha\beta}, \qquad \accentset{\circ}{\Gamma}^\alpha_{ \beta r} = r^{-1} \delta^\alpha_\beta, \qquad \accentset{\circ}{\Gamma}^\alpha_{\beta\gamma} = (\Gamma_{\sigma})^\alpha_{\beta\gamma},
\end{equation}
where $\accentset{\circ}{\Gamma}_{ij}^{l}$ and $(\Gamma_{\sigma})^\alpha_{\beta\gamma}$ are Christoffel symbols for the metrics $\delta$ and $\sigma$ respectively, and
hence
\begin{align}
\begin{split}
(\operatorname{div}_\delta g) (\nu)
& = (\operatorname{div}_\delta g)(\partial_{r}) \\
& = \accentset{\circ}{\nabla}_r g_{rr} + r^{-2} \sigma^{\alpha\beta} \accentset{\circ}{\nabla}_\beta g_{\alpha r} \\
& = \partial_r g_{rr} - 2 \accentset{\circ}{\Gamma}^l_{rr} g_{lr} +  r^{-2} \sigma^{\alpha\beta}(\partial_\beta g_{\alpha r}                                            - \accentset{\circ}{\Gamma}^l_{\alpha\beta} g_{lr}  - \accentset{\circ}{\Gamma}^l_{\beta r} g_{\alpha l} ) \\
& = - r^{-2} \sigma^{\alpha\beta} g_{\alpha \gamma} \accentset{\circ}{\Gamma}^\gamma_{\beta r} + O(r^{-3}) \\
& = -2r^{-1} + O(r^{-3}).
\end{split}
\end{align}
Furthermore
\begin{align}
\begin{split}
\operatorname{div}_\delta (u^2 df^2) (\nu)
& = \operatorname{div}_\delta (u^2 df^2)(\partial_{r}) \\
& = \accentset{\circ}{\nabla}_r (u^2 f_r^2) + r^{-2} \sigma^{\alpha\beta} \accentset{\circ}{\nabla}_\beta(u^2 f_\alpha f_r) \\
& = 2 u u_r f_r^2  + 2 u^2 f_r \accentset{\circ}{\nabla}_{rr} f + r^{-2} \sigma^{\alpha\beta} (2u u_\beta f_\alpha f_r + u^2 f_r \accentset{\circ}{\nabla}_{\alpha\beta} f + u^2 f_\alpha \accentset{\circ}{\nabla}_{\beta r} f),
\end{split}
\end{align}
with
\begin{equation}
 2 u u_r f_r^2 + 2 u^2 f_r \accentset{\circ}{\nabla}_{rr} f= -2 (u_0 + \mathcal{A}) r^{-2} + O(r^{-3+\varepsilon}),
\end{equation}
\begin{equation}
r^{-2} \sigma^{\alpha\beta} (2u u_\beta f_\alpha f_r + u^2 f_\alpha \accentset{\circ}{\nabla}_{\beta r} f) = - u^2 r^{-2} \sigma^{\alpha\beta}f_\alpha \accentset{\circ}{\Gamma}^\gamma_{\beta r} f_\gamma + O(r^{-3}) = O(r^{-3}),
\end{equation}
and
\begin{align}
\begin{split}
 u^2 r^{-2} \sigma^{\alpha\beta} f_r \accentset{\circ}{\nabla}_{\alpha\beta} f
& = u^2 r^{-2} \sigma^{\alpha\beta} f_r \accentset{\circ}{\nabla}_{\alpha\beta} \mathcal{B}
+ u^2 r^{-2} \sigma^{\alpha\beta} f_r \accentset{\circ}{\nabla}_{\alpha\beta} (f - \mathcal{B})\\
& = u^2 r^{-2}f_r\Delta_\sigma \mathcal{B} - u^2 r^{-2} \sigma^{\alpha\beta} f_r \accentset{\circ}{\Gamma}^r _{\alpha\beta} (f - \mathcal{B})_r +  O(r^{-3+\varepsilon})\\
& = r^{-2} \Delta_\sigma \mathcal{B} +  2 u^2 r^{-1} (f_r)^2  +  O(r^{-3+\varepsilon}) \\
& = r^{-2} \Delta_\sigma \mathcal{B} + 2 r^{-1} + 4 (u_0 + \mathcal{A}) r^{-2} + O(r^{-3+\varepsilon}).
 \end{split}
\end{align}
Finally, we have
\begin{align}
\begin{split}
 (d Tr_{\delta} \overline{g}) (\nu) & = \partial_r (\overline{g}_{rr} + r^{-2} \sigma^{\alpha \beta} \overline{g}_{\alpha \beta}) \\
                                                 & = \partial_r(g_{rr} + u^2 f_r^2 + r^{-2} \sigma^{\alpha\beta} (g_{\alpha\beta} + u^2 f_\alpha f_{\beta}))\\
																								 & = \partial_r (u^2 f_r^2) + O(r^{-3})\\
																								 & = -2 (u_0 + \mathcal{A})r^{-2} + O(r^{-3 + \varepsilon}).
\end{split}
\end{align}
Summing up, we conclude that
\begin{equation}
\overline{m}_{adm}=\frac{1}{16\pi}\int_{S_{\infty}}\left[(\Delta_\sigma \mathcal{B} + 4 u_0 + 4 \mathcal{A})r^{-2} +  O(r^{-3+\varepsilon})\right] = u_0 + \mathcal{A} = 2m + u_0.
\end{equation}
\end{proof}

Let $\{\overline{S}_{\tau}\}$ be an inverse mean curvature flow (IMCF) inside the Jang graph starting at the minimal boundary $\overline{S}_{0}=\partial\overline{M}$. If we choose
\begin{equation}\label{16}
u=\sqrt{\frac{|\overline{S}_{\tau}|}{16\pi}}\overline{H}_{\tau},
\end{equation}
where $|\overline{S}_{\tau}|$ and $\overline{H}_{\tau}$ denote area and mean curvature of $\overline{S}_\tau$, then the arguments in \cite{BrayKhuri1}, \cite{BrayKhuri2} imply that
\begin{equation}\label{17}
M_{H}(\infty)-\sqrt{\frac{A}{16\pi}}\geq
M_{H}(\infty)-M_{H}(0)\geq-\frac{1}{8\pi}\int_{\overline{S}_{0}\cup\overline{S}_{\infty}}
u\overline{g}(q,\nu_{\overline{g}})
\end{equation}
where $M_{H}$ denotes Hawking mass and $\nu_{\overline{g}}$ is the unit outer normal with respect to
$\overline{g}$. Notice that by choosing $u$ as in \eqref{16}, the generalized Jang equation is coupled to the inverse mean curvature flow, and we will refer to this set of equations as the Jang-IMCF system. All of this leads to the following theorem, which generalizes the results of
\cite{BrayKhuri1}, \cite{BrayKhuri2} to the asymptotically hyperboloidal case.

\begin{theorem}\label{thm1}
Let $(M, g, k)$ be a $3$-dimensional, asymptotically hyperboloidal initial data set with a
connected outermost apparent horizon boundary, and satisfying the dominant energy condition $\mu\geq|J|$.
If the coupled Jang-IMCF system of equations admits a solution satisfying the asymptotics \eqref{11}-\eqref{13},
with a weak IMCF (in the sense of \cite{HuiskenIlmanen}), and such that the boundary of the Jang surface is minimal, then \eqref{14} holds and if equality is achieved then the initial data
arise from an embedding into the Schwarzschild spacetime.
\end{theorem}

\begin{remark}
This theorem may be generalized to the case of multiple black holes, by coupling the generalized Jang equation to Bray's conformal flow \cite{Bray}. Such a procedure has been described in detail for the asymptotically flat case in \cite{HanKhuri2}.
Moreover, the ``only if" part of the case of equality is not included in the statement above (or any of the theorems in later sections), as a consequence of the difference between the mass and energy at null infinity.
\end{remark}

\begin{proof}
The (weak) IMCF becomes smooth for sufficiently large times, and approximates coordinate spheres in the asymptotically flat end. This result was established by Huisken and Ilmanen \cite{HuiskenIlmanen1} for flows in Euclidean space, and they announced (in the same paper) that  these results hold more generally in the asymptotically flat setting.
From this one may obtain asymptotic expansions for the area and mean curvature of the flow surfaces to show that
\begin{equation}\label{3.19}
M_{H}(\infty)=\overline{m}_{adm},
\end{equation}
and with the help of \eqref{15.3}
\begin{align}\label{18}
\begin{split}
M_{H}(\infty)&=\lim_{\tau\rightarrow\infty}\sqrt{\frac{|\overline{S}_{\tau}|}{16\pi}}
\left(1-\frac{1}{16\pi}\int_{\overline{S}_{\tau}}\overline{H}_{\tau}^{2}\right)\\
&=\lim_{\tau\rightarrow\infty}\sqrt{\frac{|\overline{S}_{\tau}|}{16\pi}}
\left(1-\frac{1}{|\overline{S}_{\tau}|}\int_{\overline{S}_{\tau}}u^{2}\right)\\
&=\lim_{r\rightarrow\infty}\left(\frac{r}{2}+O(1)\right)
\left[1-\left(1+\frac{2u_{0}}{r}+O(r^{-2})\right)\right]\\
&=-u_{0}.
\end{split}
\end{align}
By Lemma \ref{lemma1} $\overline{m}_{adm}=2m+u_{0}$ and so $m=-u_{0}$, or rather $M_{H}(\infty)=m$.

We note that it is only necessary to establish $M_{H}(\infty)\leq m$ in order to achieve \eqref{14}, and this may be proven without \eqref{3.19}. To see this, recall that from \cite{HuiskenIlmanen} we have $M_{H}(\infty)\leq\overline{m}_{adm}$, so \eqref{18} implies that $-u_{0}\leq\overline{m}_{adm}$. Since $\overline{m}_{adm}=2m+u_{0}$, it follows that $-\overline{m}_{adm}\leq-m$. Therefore $-u_{0}=2m-\overline{m}_{adm}\leq m$.

Consider now the boundary terms of \eqref{17}. As in \cite{BrayKhuri1}, \cite{BrayKhuri2} the inner term vanishes since $u=0$ on $\overline{S}_{0}$. Furthermore, it is shown in the Appendix B that the term at null infinity also vanishes as a result of the asymptotics \eqref{15.3}, \eqref{11}, and with the help of \eqref{3.19}.
The desired inequality now follows, and the case of equality may be treated in the same way as in \cite{BrayKhuri1}, \cite{BrayKhuri2}.
\end{proof}

In order to lend further credence to the above procedure we show that solutions to the generalized Jang equation
exist with the desired asymptotics at the horizon and at null infinity. For this we assume that $(g,k)$ takes the form \eqref{1} with $a$ and $b$ as in \eqref{2} and \eqref{3} additionally satisfying $a_{rr} = a_{r\alpha} = 0$; in this case $\mathbf{m}^{r}=0$. We note that given sufficiently regular asymptotically hyperboloidal initial data $(g,k)$ such that $|g - g_0|_{g_{0}} = O(r^{-3})$ and $|k - g_0|_{g_{0}} = O(r^{-3})$, one can perform a change of coordinates at infinity as described in the Appendix A to achieve $a_{rr} = a_{r\alpha} = 0$, and this change of coordinates does not affect the mass aspect function.

\begin{theorem}\label{thm2}
Given a smooth positive function $u$, vanishing on $\partial M$ and satisfying \eqref{15.3}, there exists a smooth solution to the generalized Jang equation \eqref{15}
which blows-up (down) at the future (past) apparent horizon boundary components and also possesses the expansion \eqref{11}.
Moreover, precise asymptotics at the horizon are given as in \cite{HanKhuri1}.
\end{theorem}

\begin{proof}
The blow-up (blow-down) and asymptotics at the horizon follow directly from the methods of \cite{HanKhuri1}. In order
to obtain the expansion \eqref{11} one may essentially follow the proof in \cite{Sakovich}
(see also \cite{SakovichThesis}), which although was designed for the classical Jang equation (that is when $u=1$), is still valid in this more general case.

The construction of barriers remains essentially the same. The idea is to look for
barriers of the form $\overline{f}(r,\theta,\phi)=\zeta(r) + \mathcal{B}(\theta,\phi)$, where $\mathcal{B}(\theta,\phi)$ is defined by \eqref{12}. In order to determine $\zeta$, it is convenient to use the substitution
\begin{equation}
p(r)=\frac{\zeta'(r)\sqrt{1+r^2}}{\sqrt{1+(1+r^2)\zeta'(r)^2}},
\end{equation}
so that $-1\leq p \leq 1$, and $p= \pm 1$ when $\zeta'=\pm \infty$. We remark that
\begin{align}\label{3.20}
\begin{split}
\beta(r,\theta,\varphi)& =\frac{1+(1+r^2)\zeta'^2}{u^{-2}+|\nabla \overline{f}|_{g}^2}\\
                       & = \frac{1}{1 + (1-p^2)(u^{-2}-1 + |\nabla \mathcal{B}|_g^2)}\\
                       & = 1 - \frac{2 \overline{m}_{adm}}{r} (1-p^2) + O(r^{-2}),
\end{split}
\end{align}
whereas the more simple expansion $\beta=1+O(r^{-2})$ holds in the setting of \cite{SakovichThesis, Sakovich}; in \eqref{3.20} the relation $u_{0}=-\overline{m}_{adm}$ is used, which follows from \eqref{3.19} and \eqref{18}. Nevertheless, a careful
computation shows that for $\overline{f}$ as above we have
\begin{align}
\begin{split}
\frac{\mathcal{F}(\overline{f})}{\beta^{3/2}\sqrt{1 + r^2}(u^{-2} + |\nabla \mathcal{B}|^2_g)}
= & p'  +\frac{2}{r}\left(p-\frac{r}{\sqrt{1+r^2}}\right)-\frac{1-p^2}{\sqrt{1+r^2}}
    -\frac{\mathcal{A} \sqrt{1-p^2}}{r^2\sqrt{1+r^2}}\\
  & + O(r^{-2})\left(\sqrt{\frac{1-p^2}{1+r^2}}-\frac{3p}{r^2} + \frac{2}{r^2} \right) + O(r^{-2}) \left( \sqrt{\frac{1-p^2}{1+r^2}} - \frac{1}{r^2} \right) \\
	& + O(r^{-2}) \left(p-\frac{r}{\sqrt{1+r^2}}\right)  + O(r^{-3}) (1-p^2) + O(r^{-2}) (1-p^2)^2  \\
  & + O(r^{-4})\sqrt{1-p^2} + O(r^{-5}),
\end{split}
\end{align}
where $\mathcal{F}(\overline{f})$ is the left-hand %portion
side of the generalized Jang equation \eqref{15} computed for $f = \overline{f}$. %which does not involve $k$.
Following \cite{SakovichThesis, Sakovich}, we define $p_+$ and $p_-$ to be solutions of the boundary value problems
\begin{equation}%\label{21}
\begin{split}
 p_{\pm}'
  & +\frac{2}{r}\left(p_{\pm}-\frac{r}{\sqrt{1+r^2}}\right)-\frac{1-p_{\pm}^2}{\sqrt{1+r^2}} -\frac{\mathcal{A} \sqrt{1-p_{\pm}^2}}{r^2\sqrt{1+r^2}}\\
  & \pm C_1 r^{-2}\left|\sqrt{\frac{1-p_{\pm}^2}{1+r^2}}-\frac{3p_{\pm}}{r^2} + \frac{2}{r^2} \right| \pm C_2 r^{-2} \left| \sqrt{\frac{1-p_{\pm}^2}{1+r^2}}
	  - \frac{1}{r^2} \right| \pm    C_3 r^{-2}\left|p_{\pm}-\frac{r}{\sqrt{1+r^2}}\right|  \\
	& \pm C_4 r^{-3}(1-p_{\pm}^2) \pm C_5 r^{-2} (1-p_{\pm}^2)^2 \pm C_6 r^{-4}\sqrt{1-p_{\pm}^2} \pm C_7 r^{-5} = 0,
\end{split}
\end{equation}

\begin{equation}%\label{22}
p_{\pm}(r_0)  = \mp 1,
\end{equation}
where $C_i$, $i=1,\ldots,7$, are positive constants. The same analysis as in \cite{SakovichThesis, Sakovich} applies to this system, the properties and the asymptotics of solutions remaining the same. This in turn gives rise to barriers $\overline{f}_\pm$ having the asymptotic expansion \eqref{11}. The rest of the proof proceeds as in \cite{SakovichThesis, Sakovich}.

%The only difficulty is
%making sure that the barriers constructed still work with the lower %order terms $\phi_{i}f_{j}$. However
%since $\phi_{i}=O(r^{-2})$, this should not be a problem. [[Of course %more details are needed here.]]
\end{proof}

\begin{remark}\label{remark001}
Since scalar multiplication is not a homothety for the hyperbolic metric, the estimate \eqref{13} cannot be obtained by applying the rescaling technique directly to the generalized Jang equation \eqref{15}. Instead, in order to show that Lemma \ref{lemma1} and Theorem \ref{thm1} hold for the metric $\overline{g} = g + u^2 df^2$ one may argue as in \cite{Sakovich} where the case $u=1$ was considered. More specifically, from the proof of Theorem \ref{thm2} we know that outside of a compact set the Jang graph $\overline{M}$ lies between the graph $M_-$ of the lower barrier $f_- : M \to \mathbb{R}$ and the graph $M_+$ of the upper barrier $f_+ :M \to \mathbb{R}$, where $f_-$ and $f_+$ have the asymptotic expansions \eqref{11}-\eqref{13}. Combining this with the $C^0$ estimate for the second fundamental form of $\overline{M}$ obtained in \cite[Theorem 2.2]{HanKhuri1} one may show that the asymptotic end of $\overline{M}$ can be viewed as the graph of the function $h : M_- \to \mathbb{R}$ in the Gaussian normal coordinates adapted to $M_- \subset M \times \mathbb{R}$. Recall \cite{BrayKhuri1, BrayKhuri2} that the generalized Jang equation tells us that the mean curvature of $\overline{M}$ is equal to the trace of a certain extension of $k$ to $M \times \mathbb{R}$ over $\overline{M}$. Expressing this in terms of $h: M_- \to \mathbb{R}$ will give a slightly more complicated equation than \eqref{15}, since the ambient metric will no longer have a warped product structure in the described coordinates. However, the induced metric on $M_-$ is asymptotically Euclidean in the sense of \eqref{9}, which allows one to apply the rescaling technique to this equation and thereby derive the analogue of estimate \eqref{13} for the function $h$. Translating this back to the setting in which $\overline{M}$ is expressed as the graph $f: M \to  \mathbb{R}$ then yields the desired estimates.
\end{remark}

\section{The Penrose Inequality with Charge}
\label{sec4} \setcounter{equation}{0}
\setcounter{section}{4}

Let $(M,g,k,E)$ be an initial data set for the Einstein-Maxwell equations as described in Section \ref{sec2},
with $E$ divergence free. For simplicity in this section, we will assume that the magnetic field vanishes $B=0$. %although similar arguments show that the main results hold in case %this field is nontrivial.
We seek a deformation of the initial data to $(\overline{M},\overline{g},\overline{E})$ such that the charged dominant energy condition holds weakly in the time symmetric case, that is $\overline{R}\geq 2|\overline{E}|_{\overline{g}}^{2}$ when integrated against an appropriate test function. Moreover, several other aspects of the geometry should be preserved, namely
\begin{equation}\label{21-0}
\overline{m}_{adm}=m,\text{ }\text{ }\text{ }\text{ }\text{ }\text{ }\overline{\mathcal{Q}}_{e}=\mathcal{Q}_{e}, \text{ }\text{ }\text{ }\text{ }\text{ }\text{ }\operatorname{div}_{\overline{g}}\overline{E}=(1+u^{2}|\nabla f|_{g}^{2})^{-1/2}\operatorname{div}_{g}E,\text{ }\text{ }\text{ }\text{ }\text{ }\text{ }|E|_{g}\geq|\overline{E}|_{\overline{g}}.
\end{equation}
This will be achieved by choosing $\overline{g}=g+u^{2}df^{2}$ where $f$ solves the generalized Jang equation
having the asymptotics and boundary behavior as in Section \ref{sec3}, and also with the same choice of warping factor \eqref{16}. In particular, combining the arguments in the previous section with ideas from \cite{DisconziKhuri} we obtain $\overline{m}_{adm}=m$. Furthermore, by choosing
\begin{equation}\label{22-0}
\overline{E}_i = \frac{E_i + u^2 f_i f^j E_j}{\sqrt{1 + u^2 |\nabla f|^2_g}}
\end{equation}
the last two properties of \eqref{21-0} are satisfied, as is shown in \cite{DisconziKhuri}. We will now show that $\overline{\mathcal{Q}}_{e}=\mathcal{Q}_{e}$. First observe that
\begin{equation}\label{4.1}
|\nabla f|_{g}^{2}\sim r^{2},\text{ }\text{ }\text{ }\text{ }\text{ }\text{ } f_{r}\sim 1, \text{ }\text{ }\text{ }\text{ }\text{ }\text{ }f^{r}\sim r^{2},\text{ }\text{ }\text{ }\text{ }\text{ }\text{ }
f^{\alpha}=O(r^{-2}),
\end{equation}
\begin{equation}\label{4.2}
\nu_{\overline{g}}^{r}\sim 1,\text{ }\text{ }\text{ }\text{ }\text{ }\text{ }
\nu_{\overline{g}}^{\alpha}=O(r^{-2}),\text{ }\text{ }\text{ }\text{ }\text{ }\text{ }
\nu_{g}^{r}\sim r,\text{ }\text{ }\text{ }\text{ }\text{ }\text{ }
\nu_{g}^{\alpha}=O(r^{-4}),
\end{equation}
and by \eqref{6}
\begin{equation}\label{4.3}
f^{j}E_{j}=f^r E_r +O(r^{-3})=O(r^{-1}).
\end{equation}
It follows that
\begin{align}\label{23}
\begin{split}
\overline{E}_{i}\nu_{\overline{g}}^{i}&=\overline{E}_{r}\nu_{\overline{g}}^{r}
+\overline{E}_{\alpha}\nu_{\overline{g}}^{\alpha}\\
&=r^{-1}(E_{r}+u^{2}f_{r}f^{j}E_{j})
+r^{-3}(E_{\alpha}+u^{2}f_{\alpha}f^{j}E_{j})+O(r^{-3})\\
&=rE_{r}+O(r^{-3})\\
&=E_{i}\nu^{i}_{g}+O(r^{-3}),
\end{split}
\end{align}
which yields the desired conclusion.

\begin{theorem}\label{thm3}
Let $(M, g, k, E)$ be a $3$-dimensional, asymptotically hyperboloidal initial data set for the Einstein-Maxwell equations with a
connected outermost apparent horizon boundary, and satisfying the charged dominant energy condition $\mu_{EM}\geq|J_{EM}|$ as well as $\operatorname{div}_{g}E=0$.
If the coupled Jang-IMCF system of equations admits a solution satisfying the asymptotics \eqref{11}-\eqref{13},
with a weak IMCF (in the sense of \cite{HuiskenIlmanen}), and such that the boundary of the Jang surface is minimal, then
\begin{equation}\label{24}
m\geq \sqrt{ \frac{A }{16\pi} } + \sqrt{ \frac{\pi}{A } } \mathcal{Q}^2,
\end{equation}
and if equality is achieved then the initial data
arise from an embedding into the Reissner-Nordstr\"{o}m spacetime.
\end{theorem}

\begin{proof}
Since \eqref{21-0} holds, the theorem follows directly from the arguments in the asymptotically flat case \cite{DisconziKhuri}.
\end{proof}

\begin{remark}
Note that the Jang-IMCF system of equations is exactly the same as in Section \ref{sec3}. Thus, the existence result Theorem \ref{thm2} provides further credence to the above procedure.
Moreover, it should be possible to generalize this result to the case of multiple black holes if an additional area-charge inequality is satisfied by the horizon as in \cite{KhuriWeinsteinYamada1}. This will require a coupling of the generalized Jang equation to the charged conformal flow \cite{KhuriWeinsteinYamada1}.
\end{remark}

\section{The Positive Mass Theorem with Charge}
\label{sec5} \setcounter{equation}{0}
\setcounter{section}{5}

Let $(M,g,k,E)$ be an initial data set for the Einstein-Maxwell equations as described in Section \ref{sec2}. We will assume for convenience, as in the previous section, that the magnetic field vanishes, but here the
electric field $E$ need not be divergence free. Again we seek a deformation of the initial data to $(\overline{M},\overline{g},\overline{E})$, where
$\overline{g}=g+u^{2}df^{2}$ and $\overline{E}$ is given in \eqref{22-0}. As before, $f$ solves the generalized Jang equation having the asymptotics as in Section \ref{sec3}, and boundary behavior (at the horizon) as described in \cite{KhuriWeinstein}. However, the warping factor is chosen differently, and this will be outlined below.
%Note that \eqref{21} holds with this choice of $u$ as long as the %expansion \eqref{19} is valid.

In order to choose $u$, we must describe the appropriate spinors on the Jang surface. Dirac spinors \cite{ParkerTaubes} are sections of the (vector) spinor bundle $\mathcal{S}$ over $\overline{M}$ with structure group $SL(2,\mathbb{C})$.  The (Jang) metric compatible connection on $\mathcal{S}$ is given by
\begin{equation}\label{25}
 \overline{\nabla}_{e_{i}}=e_{i}+\frac{1}{4}\, \overline{\omega}_{ijl}\,
 e_{j}\cdot e_{l}\cdot
\end{equation}
where $\overline{\omega}_{ijl}=\overline{g}(\overline{\nabla}_{e_i} e_j,e_l)$ are connection coefficients associated with an orthonormal frame field
$(e_1,e_2,e_3)$, and %$c: T^* \overline{M}\rightarrow \mathrm{End}(\mathcal{S})$
$\cdot$ indicates
Clifford multiplication. The Einstein-Maxwell spin connection on $\mathcal{S}$, which is relevant for the positive mass theorem with charge \cite{GHHP}, then has the form
\begin{equation}\label{27}
   \nabla_{e_{i}}= \overline{\nabla}_{e_i}
   -\frac{1}{2} \, \overline{E}\cdot e_{i}\cdot e_{0}\cdot
\end{equation}
where $e_0$ is the unit normal to $\overline{M}$ in the (Lorentzian) warped product 4-manifold described in Section \ref{sec3}. Observe that this connection is not metric compatible due to the contribution of the electric field. Let $\Gamma(\mathcal{S})$ be the space of cross-sections, then the Einstein-Maxwell Dirac operator
$\slashed{D} : \Gamma(\mathcal{S})\rightarrow\Gamma(\mathcal{S})$ is  defined by
\begin{equation}\label{28}
  \slashed{D}\psi = \sum_{i=1}^3 e_i\cdot
  \nabla_{e_i} \psi,
\end{equation}
and a spinor $\psi$ on $\overline{M}$ is called harmonic if it satisfies the Dirac equation
\begin{equation}\label{29}
\slashed{D} \psi = 0.
\end{equation}
The Dirac equation is coupled to the generalized Jang equation through the choice
\begin{equation}\label{30}
u=|\psi|^{2}.
\end{equation}

\begin{lemma}\label{lemma2}
Fix a complete asymptotically flat initial data set $(\overline{M},\overline{g},\overline{E})$, with asymptotically cylindrical ends, and satisfying $|\overline{R}|+|\operatorname{div}_{\overline{g}}\overline{E}|=o(r^{-3})$ as $r \rightarrow \infty$.
Let $\psi$ solve the Dirac equation \eqref{29} with $\psi\rightarrow\psi_{0}$ in the asymptotic end, where $\psi_{0}$ is a constant spinor of modulus 1. Then $u$ as defined by \eqref{30} has the asymptotic expansion
\begin{equation}\label{5.1}
u= 1 + \frac{u_{0}}{r}+O_{2}\left(\frac{1}{r^{2-\epsilon}}\right),
\end{equation}
for any $\epsilon>0$. Moreover $u_{0}=-2\overline{m}_{adm}-\overline{\mathcal{Q}}_{e}\langle\psi_{0},e_{0}\cdot\psi_{0}\rangle$.
\end{lemma}

\begin{proof}
Observe that
\begin{equation}\label{11111}
\Delta_{\overline{g}}|\psi|^{2}=\overline{\nabla}_{e_{i}}\overline{\nabla}_{e_{i}}|\psi|^{2}
-\overline{\nabla}_{\overline{\nabla}_{e_{i}}e_{i}}|\psi|^{2}
=\langle\overline{\nabla}_{e_{i},e_{i}}^{2}\psi,\psi\rangle
+\langle\psi,\overline{\nabla}_{e_{i},e_{i}}^{2}\psi\rangle
+2|\overline{\nabla}\psi|^{2},
\end{equation}
where $\overline{\nabla}_{e_{i},e_{i}}^{2}=\overline{\nabla}_{e_{i}}\overline{\nabla}_{e_{i}}
-\overline{\nabla}_{\overline{\nabla}_{e_{i}}e_{i}}$ is the connection Laplacian. In what follows, calculations will be performed at a point where the orthonormal frame has been chosen such that $\overline{\nabla}_{e_{i}}e_{j}=0$; note that we also have $\overline{\nabla}_{e_{i}}e_{0}=0$ as the $t=0$ slice is totally geodesic in the Lorentzian setting of the Jang deformation. Moreover, for simplicity $\overline{\nabla}_{e_{i}}$ will be denoted by $\overline{\nabla}_{i}$. If $A_{i}=\frac{1}{2}\overline{E}\cdot e_{i}\cdot e_{0}$ then
\begin{align}
\begin{split}
0=\slashed{D}^{2}\psi=&e_{i}\cdot\left(\overline{\nabla}_{i}-A_{i}\cdot\right)
e_{j}\cdot\left(\overline{\nabla}_{j}-A_{j}\cdot\right)\psi\\
=&e_{i}\cdot\overline{\nabla}_{i}\left(e_{j}\cdot\overline{\nabla}_{j}\psi\right)
-e_{i}\cdot\overline{\nabla}_{i}\left(e_{j}\cdot A_{j}\cdot\psi\right)
-e_{i}\cdot A_{i}\cdot e_{j}\cdot\overline{\nabla}_{j}\psi
+e_{i}\cdot A_{i}\cdot e_{j}\cdot A_{j}\cdot\psi\\
=&-\overline{\nabla}_{i}\overline{\nabla}_{i}\psi+\frac{1}{4}\overline{R}\psi
-e_{i}\cdot e_{j}\cdot A_{j}\cdot\overline{\nabla}_{i}\psi
-\frac{1}{2}e_{i}\cdot e_{j}\cdot\overline{\nabla}_{i}\overline{E}\cdot
e_{j}\cdot e_{0}\cdot\psi,
\end{split}
\end{align}
after applying the Lichnerowicz-Weitzenb\"{o}ck formula. It follows that
\begin{align}
\begin{split}
\langle\overline{\nabla}_{i}\overline{\nabla}_{i}\psi,\psi\rangle
+\langle\psi,\overline{\nabla}_{i}\overline{\nabla}_{i}\psi\rangle
=&\frac{1}{2}\overline{R}|\psi|^{2}
-\langle e_{i}\cdot e_{j}\cdot A_{j}\cdot\overline{\nabla}_{i}\psi,\psi\rangle
-\langle\psi,e_{i}\cdot e_{j}\cdot A_{j}\cdot\overline{\nabla}_{i}\psi\rangle\\
&-\frac{1}{2}(\langle e_{i}\cdot e_{j}\cdot\overline{\nabla}_{i}\overline{E}\cdot
e_{j}\cdot e_{0}\cdot\psi,\psi\rangle
+\langle\psi,e_{i}\cdot e_{j}\cdot\overline{\nabla}_{i}\overline{E}\cdot
e_{j}\cdot e_{0}\cdot\psi\rangle).
\end{split}
\end{align}
Moreover a computation shows that
\begin{equation}
\langle e_{i}\cdot e_{j}\cdot\overline{\nabla}_{i}\overline{E}\cdot
e_{j}\cdot e_{0}\cdot\psi,\psi\rangle
+\langle\psi,e_{i}\cdot e_{j}\cdot\overline{\nabla}_{i}\overline{E}\cdot
e_{j}\cdot e_{0}\cdot\psi\rangle
=-2(\operatorname{div}_{\overline{g}}\overline{E})\langle\psi,e_{0}\cdot\psi\rangle,
\end{equation}
and hence
\begin{equation}\label{u equation}
\Delta_{\overline{g}}|\psi|^{2}=2|\overline{\nabla}\psi|^{2}
+\frac{1}{2}\overline{R}|\psi|^{2}
-\langle e_{i}\cdot e_{j}\cdot A_{j}\cdot\overline{\nabla}_{i}\psi,\psi\rangle
-\langle\psi,e_{i}\cdot e_{j}\cdot A_{j}\cdot\overline{\nabla}_{i}\psi\rangle
+(\operatorname{div}_{\overline{g}}\overline{E})\langle\psi,e_{0}\cdot\psi\rangle.
\end{equation}

According to \cite{ParkerTaubes}, in the asymptotically flat end $\psi=\psi_{0}+O(r^{-1+\epsilon/2})$ and $|\overline{\nabla}\psi|=O(r^{-2+\epsilon/2})$ for all $\epsilon>0$. This, combined with the assumption $|\overline{R}|+|\operatorname{div}_{\overline{g}}\overline{E}|=o(r^{-3})$ implies that
\begin{equation}
\Delta_{\overline{g}}|\psi|^{2}=o(r^{-3}).
\end{equation}
It follows that $u$ has the desired expansion (see for instance \cite{SmithWeinstein}).

Lastly, we compute the value of $u_{0}$. Observe that \eqref{11111} may be rewritten as
\begin{equation}
\Delta_{\overline{g}}u=\operatorname{div}_{\overline{g}}
\left(\langle\overline{\nabla}_{\bullet}\psi,\psi\rangle\right)
+\operatorname{div}_{\overline{g}}
\left(\langle\psi,\overline{\nabla}_{\bullet}\psi\rangle\right).
\end{equation}
Integrating by parts produces
\begin{align}\label{000000}
\begin{split}
-4\pi u_{0}=&\int_{\overline{S}_{\infty}}\langle\overline{\nabla}_{\nu_{\overline{g}}}\psi,\psi\rangle
+\langle\psi,\overline{\nabla}_{\nu_{\overline{g}}}\psi\rangle\\
=&\int_{\overline{S}_{\infty}}\langle\overline{\nabla}_{\nu_{\overline{g}}}\psi
+\nu_{\overline{g}}\cdot\overline{\slashed{D}}\psi,\psi\rangle
+\langle\psi,\overline{\nabla}_{\nu_{\overline{g}}}\psi
+\nu_{\overline{g}}\cdot\overline{\slashed{D}}\psi\rangle\\
&-\int_{\overline{S}_{\infty}}\left(\langle
\nu_{\overline{g}}\cdot\overline{\slashed{D}}\psi,\psi\rangle
+\langle\psi,
\nu_{\overline{g}}\cdot\overline{\slashed{D}}\psi\rangle\right)\\
=&8\pi\overline{m}_{adm}-\int_{\overline{S}_{\infty}}\left(\langle
\nu_{\overline{g}}\cdot\overline{\slashed{D}}\psi,\psi\rangle
+\langle\psi,
\nu_{\overline{g}}\cdot\overline{\slashed{D}}\psi\rangle\right).
\end{split}
\end{align}
Note that the asymptotically cylindrical ends do not contribute to a boundary integral, since in these regions $u$ decays exponentially. Furthermore
\begin{equation}
\overline{\slashed{D}}\psi=e_{i}\cdot\overline{\nabla}_{i}\psi
=e_{i}\cdot(\nabla_{i}+A_{i}\cdot)\psi=e_{i}\cdot A_{i}\cdot\psi,
\end{equation}
so that
\begin{equation}\label{123450}
\langle
\nu_{\overline{g}}\cdot\overline{\slashed{D}}\psi,\psi\rangle
+\langle\psi,
\nu_{\overline{g}}\cdot\overline{\slashed{D}}\psi\rangle=
-\overline{g}(\overline{E},\nu_{\overline{g}})\langle\psi,e_{0}\cdot\psi\rangle.
\end{equation}
The desired result follows by combining \eqref{000000} and \eqref{123450}.
\end{proof}

\begin{theorem}\label{thm4}
Let $(M, g, k, E)$ be a $3$-dimensional, asymptotically hyperboloidal initial data set for the Einstein-Maxwell equations with an
outermost apparent horizon boundary, and satisfying the charged dominant energy condition \eqref{10}.
If the coupled Dirac-Jang system of equations admits a solution satisfying the asymptotics \eqref{11}-\eqref{13}, \eqref{5.1}, and such that the Jang surface possesses an asymptotically cylindrical neck over the horizon, then
\begin{equation}\label{24}
m\geq|\mathcal{Q}|,
\end{equation}
and if equality is achieved then the initial data arise from an embedding into the Majumdar-Papapetrou spacetime.
\end{theorem}

\begin{proof}
We have that \eqref{21-0} is valid, except for the equality of masses. Thus, we may follow the arguments
in \cite{KhuriWeinstein} and use the Lichnerowicz-Weitzenb\"{o}ck formula to obtain
\begin{equation}\label{5.3}
-\frac{1}{2}\int_{\overline{S}_{\infty}}u\overline{g}(q,\nu_{\overline{g}})\leq \int_{\overline{M}}|\nabla\psi|^{2}+\frac{1}{4}\langle\psi,
\mathcal{R}\cdot\psi\rangle=4\pi(\overline{m}_{adm}+\overline{\mathcal{Q}}_{e}\langle\psi_{0},
e_{0}\cdot\psi_{0}\rangle),
\end{equation}
where $\mathcal{R}=\overline{R}-2|\overline{E}|_{\overline{g}}^{2}
-(\operatorname{div}_{\overline{g}}\overline{E})e_{0}$. Note that there is no interior boundary integral on the left-hand side of \eqref{5.3}; as in Lemma \ref{lemma2} this is due to the fact that the Jang surface has asymptotically cylindrical ends, which yields exponential decay of $u$ in these regions.
By Lemma \ref{lemma3} and $\overline{\mathcal{Q}}_{e}=\mathcal{Q}_{e}$, \eqref{5.3} becomes
\begin{equation}
u_{0}+m\leq\overline{m}_{adm}-|\mathcal{Q}_{e}|.
\end{equation}
This produces the desired result, since $\overline{m}_{adm}=2m+u_{0}$ according to Lemma \ref{lemma1}. Lastly, the case of equality is treated the same way as in \cite{KhuriWeinstein}.
\end{proof}

\begin{remark}
It should be pointed out that slices of the Majumdar-Papapetrou spacetime which agree with a $t=const$ slice near the horizon do not fall under the hypotheses of this theorem, as such initial data possess asymptotically cylindrical ends. This is related to the fact that in order to obtain asymptotically cylindrical ends in the Jang surface, when the initial data have an apparent horizon boundary, use of blow-up solutions of the generalized Jang equation is required. Thus, more general types of boundary behavior for the generalized Jang equation are needed, if this theorem is to allow
initial data with asymptotically cylindrical ends. Note also that the existence result Theorem \ref{thm2} holds in this setting, and provides further credence to the above procedure.

In order to satisfy the hypotheses of Lemma \ref{lemma2} concerning the fall-off of the scalar curvature and divergence of the electric field, it may be necessary to assume $\mu+|J|_{g}=O(r^{-3-\varepsilon})$ for $\varepsilon>0$. This extra condition together with the charged dominant energy condition \eqref{10} and the relation between divergences in \eqref{21-0} implies that $|\operatorname{div}_{\overline{g}}\overline{E}|=O(r^{-3-\varepsilon})$. Moreover, this condition also guarantees that $|\overline{R}|=O(r^{-3-\varepsilon})$ modulo a divergence term in light of the formula \eqref{15.1}. Due to the fall-off of the divergence term, $u_{0}$ in the expansion \eqref{5.1} may be a nonconstant function on the 2-sphere. This function can be determined from asymptotic expansions as follows. A computation shows that
\begin{equation}
\overline{R}=\frac{2\Delta_{\sigma}(\mathcal{B}-u_{0})}{r^{3}}+O(r^{-3-\varepsilon})
\end{equation}
where $\mathcal{B}$ is given in \eqref{12}, and so from \eqref{u equation}
\begin{equation}
\Delta_{\overline{g}}u=\frac{\Delta_{\sigma}(\mathcal{B}-u_{0})}{r^{3}}+O(r^{-3-\varepsilon}).
\end{equation}
It now follows from the expansion \eqref{5.1} that
\begin{equation}
\Delta_{\sigma}u_{0}=\Delta_{\sigma}(\mathcal{B}-u_{0}),
\end{equation}
and hence $2u_{0}=\mathcal{B}+const$. Lastly, it should be added that all the arguments of the current section go through with this version of $u_{0}$ after making suitable modifications. Ultimately, the validity of the expansions will only be determined once the coupled system is solved. Here we are simply pointing out that the canonical expansions are consistent with one another.
\end{remark}

\section{The Mass-Angular Momentum Inequality}
\label{sec6} \setcounter{equation}{0}
\setcounter{section}{6}

Let $(M,g,k)$ be an initial data set as described in Section \ref{sec2} with
\begin{equation} \label{31}
b_{rr} = \frac{b^{(5)}_{rr} (\theta, \phi)}{r^{5}} + O_{2}(r^{-6})
\end{equation}
where $b^{(5)}_{rr} (\theta, \phi)$ is a function on $S^{2}$.
In this section we assume further that $M$ is simply connected and that there are two ends, one denoted $M_{end}^{+}$ which is asymptotically hyperboloidal, and the other $M_{end}^{-}$ which is either asymptotically flat or asymptotically cylindrical.  It is also assumed that the initial data are axisymmetric; without this assumption the mass-angular momentum inequality is no longer generally valid \cite{HuangSchoenWang}. By this we mean that
there is a subgroup
isomorphic to $U(1)$ contained within the group of isometries of the Riemannian manifold $(M,g)$, and that all quantities associated with the initial data are invariant under the $U(1)$ action. In particular, if $\eta=\partial_{\phi}$ denotes the Killing field which generates the symmetry, then
\begin{equation}\label{32}
\mathfrak{L}_{\eta}g=\mathfrak{L}_{\eta}k=0,
\end{equation}
where $\mathfrak{L}_{\eta}$ is Lie differentiation.
In the current setting, simple connectivity and axial symmetry imply \cite{Chrusciel} that $M$ is diffeomorphic to $\mathbb{R}^{3}\setminus\{0\}$, where the origin represents a black hole and the neighboring geometry has the structure of an asymptotically flat or cylindrical end.
The mass of the asymptotically hyperboloidal end will be denoted by $m$ as usual, and the angular momentum is defined by
\begin{equation}\label{33}
\mathcal{J}=\frac{1}{8\pi}\int_{S}(k_{ij}-(Tr_{g} k)g_{ij})\nu^{i}_{g}\eta^{j},
\end{equation}
where $S$ is any surface enclosing the origin, with unit outer normal $\nu_{g}$. In order for this to be well-defined (independent of the choice of $S$), it is assumed that
\begin{equation}\label{34}
J_{i}\eta^{i}=0,
\end{equation}
which yields conservation of angular momentum \cite{DainKhuriWeinsteinYamada}; observe that \eqref{33} is equivalent to the limit definition which is also valid in the asymptotically flat setting
\begin{equation}\label{33.1}
\mathcal{J}=\frac{1}{8\pi}\int_{S_{\infty}}(k_{ij}-(Tr_{g} k)g_{ij})\nu^{i}_{g}\eta^{j}.
\end{equation}
The mass-angular momentum inequality states that
$m\geq\sqrt{|\mathcal{J}|}$.

We seek a deformation of the initial data $(M,g,k)\rightarrow(\overline{M},\overline{g},\overline{k})$ such that the manifolds are diffeomorphic $M\cong\overline{M}$, the geometry of the end $M_{end}^{-}$ is preserved while the end $M_{end}^{+}$ becomes asymptotically flat, and
\begin{equation}\label{36}
\overline{m}_{adm}=m,\text{ }\text{ }\text{ }\text{ }\text{ }\text{ }\overline{\mathcal{J}}=\mathcal{J},\text{ }\text{ }\text{ }\text{ }\text{ }\text{ }\overline{J}(\eta)=0,\text{ }\text{ }\text{ }\text{ }\text{ }\text{ }Tr_{\overline{g}}\overline{k}=0,\text{ }\text{ }\text{ }\text{ }\text{ }\text{ }
\overline{R}\geq|\overline{k}|_{\overline{g}}^{2}\text{ }\text{ }\text{ weakly}.
\end{equation}
Here $\overline{J}$ and $\overline{\mathcal{J}}$ are the momentum density and angular momentum of the new data; the later is well-defined as in \eqref{33} since \eqref{34} holds after the deformation. With intuition from other Jang-type reduction procedures, and the fact that the Kerr spacetime is stationary, we search for a graph inside a stationary 4-manifold
\begin{equation}\label{37}
\overline{M}=\{t=f(x)\}\subset(M\times\mathbb{R}, g+2Y_{i}dx^{i}dt+(u^{2}-|Y|_{\overline{g}}^{2}) dt^{2}),
\end{equation}
where all quantities are independent of $t$ and are axisymmetric
\begin{equation}\label{38}
\mathfrak{L}_{\eta}f=\mathfrak{L}_{\eta}u=\mathfrak{L}_{\eta}Y=0.
\end{equation}
Let $\overline{g}$ be the induced metric on the graph and $\overline{k}$ be the second fundamental form of the $t=0$ slice in the dual Lorentzian setting \cite{ChaKhuri1}, that is
\begin{equation}\label{39}
\overline{g}_{ij}=g_{ij}+f_{i}Y_{j}+f_{j}Y_{i}+(u^{2}-|Y|_{\overline{g}}^{2}) f_{i}f_{j},\text{ }\text{ }\text{ }\text{ }\text{ }\text{ }\text{ }
\overline{k}_{ij}=\frac{1}{2u}\left(\overline{\nabla}_{i}Y_{j}+\overline{\nabla}_{j}Y_{i}\right),
\end{equation}
where $\overline{\nabla}$ is the Levi-Civita connection with respect to $\overline{g}$. Moreover, the
structure of the Kerr spacetime suggests that we make the following simplifying ansatz that $Y$ has a single component
\begin{equation}\label{41}
\overline{Y}^{i}\partial_{i}:=\overline{g}^{ij}Y_{j}\partial_{i}=Y^{\phi}\partial_{\phi}.
\end{equation}
Thus, the deformation is defined by three functions $(u,Y^{\phi},f)$. In addition, \eqref{41} implies \cite{ChaKhuri1} that $\overline{g}$ is Riemannian and $Tr_{\overline{g}}\overline{k}=0$.

We will now show how to choose the three functions $(u,Y^{\phi},f)$. First, in order to have a well-defined angular momentum, and the  existence of a twist potential which is needed to apply techniques from the maximal case \cite{Dain0}, \eqref{34} must hold for the new data, or equivalently
\begin{equation}\label{42}
\operatorname{div}_{\overline{g}}\overline{k}(\eta)=0.
\end{equation}
As is shown in \cite{ChaKhuri1}, this is a linear elliptic equation for $Y^{\phi}$ (if $u$ is independent of $Y^{\phi}$), which has a unique bounded solution if the $r^{-3}$-fall-off rate is prescribed at $M_{end}^{+}$. Therefore we will require the expansion%\footnote{The %notation $h=o_{l}(r^{-a})$ asserts that
%$\lim_{r\rightarrow\infty}r^{a+j}\partial^{j}h=0$
%for all $j\leq l$, and
%$h=O_{l}(r^{-a})$ asserts that $r^{a+j}|\partial^{j}h|\leq C$ for all %$j\leq l$.}
\begin{equation}\label{43}
Y^{\phi}=-\frac{2\mathcal{J}}{r^{3}}+o_{2}(r^{-\frac{7}{2}})\text{ }\text{ }\text{ as }\text{ }\text{ }r\rightarrow\infty,
\end{equation}
which also ensures that $\overline{\mathcal{J}}=\mathcal{J}$. Requiring the solution to be bounded implies the following asymptotics at the other end
\begin{equation}\label{00133}
Y^{\phi}=\mathcal{Y}+O_{1}(r^{5})\text{ }\text{ }\text{ in asymptotically flat }\text{ }\text{ }M_{end}^{-},
\end{equation}
\begin{equation}\label{00134}
Y^{\phi}=\mathcal{Y}+O_{1}(r)\text{ }\text{ }\text{ in asymptotically cylindrical }\text{ }\text{ }M_{end}^{-},
\end{equation}
where $\mathcal{Y}$ is a constant determined by the data. Here $r=\sqrt{\rho^{2}+z^{2}}$ is the `Euclidean distance' from the origin in a Brill (or cylindrical) coordinate system $(\rho,\phi,z)$ in this end \cite{Chrusciel}.

Next we choose $f$ to satisfy the Jang-type equation
\begin{equation}\label{48.1}
g^{ij}\left(\frac{u\nabla_{ij}f+u_{i}f_{j}+u_{j}f_{i}}{\sqrt{1+u^{2}|\nabla f|_{g}^{2}}}-k_{ij}\right)=0\text{ }\text{ }\text{ }\text{ }\text{ }\Leftrightarrow\text{ }\text{ }\text{ }\text{ }\text{ }
\operatorname{div}_{g}(u^{2}\nabla f)=u(Tr_{g}k)\sqrt{1+u^{2}|\nabla f|_{g}^{2}}.
\end{equation}
As with the deformations of previous sections, the purpose of this equation is to impart positivity properties to the
scalar curvature. In particular, it is shown in \cite{ChaKhuri1} that
\begin{equation}\label{49}
\overline{R}-|\overline{k}|_{\overline{g}}^{2}= 2(\mu-J(v))+|k-\pi|_{g}^{2}+2u^{-1}\operatorname{div}_{\overline{g}}(uQ),
\end{equation}
where
\begin{equation}\label{45}
\pi_{ij}=\frac{u\nabla_{ij}f+u_{i}f_{j}+u_{j}f_{i}+\frac{1}{2}(g_{i\phi}Y^{\phi}_{,j}+g_{j\phi}Y^{\phi}_{,i})}{\sqrt{1+u^{2}|\nabla f|_{g}^{2}}}
\end{equation}
is the second fundamental form of the graph in the Lorentzian setting,
\begin{equation}\label{46}
v^{i}=\frac{uf^{i}}{\sqrt{1+u^{2}|\nabla f|_{g}^{2}}},\text{ }\text{ }\text{ }\text{ }\text{ }\text{ }\text{ }w^{i}=\frac{uf^{i}+u^{-1}\overline{Y}^{i}}{\sqrt{1+u^{2}|\nabla f|_{g}^{2}}},
\end{equation}
and
\begin{equation}\label{47}
Q_{i}=\overline{Y}^{j}\overline{\nabla}_{ij}f-u\overline{g}^{jl}f_{l}\overline{k}_{ij}+w^{j}(k-\pi)_{ij}+uf_{i}w^{l}w^{j}(k-\pi)_{lj}\sqrt{1+u^{2}|\nabla f|_{g}^{2}}.
\end{equation}
If the dominant energy condition is valid, it follows that $\overline{R}\geq|\overline{k}|_{\overline{g}}^{2}$ weakly in the sense that this inequality holds after multiplying by $u$ and integrating by parts (it will be shown in Appendix B that the boundary terms vanish).

Before explaining how to choose $u$, we will record the asymptotics which allow an appropriate solution of the equation \eqref{48.1}, namely
\begin{equation}\label{52}
u= 1 + \frac{\mathcal{C}_{1}}{r} + \frac{\mathcal{C}_{2}(\theta, \phi)}{r^{2}} + \frac{\mathcal{C}_{3}(\theta, \phi)}{r^{3}} + O_{2}(r^{-4})\text{ }\text{ as }\text{ }r\rightarrow\infty\text{ }\text{ in }\text{ }\text{ }M_{end}^{+},
\end{equation}
where $\mathcal{C}_{1} = -\overline{m}_{adm}$ and $\mathcal{C}_{2}, \mathcal{C}_{3}$ are functions on the sphere $S^{2}$. At the other end $M_{end}^{-}$ the asymptotics are required to be
\begin{equation}\label{53}
u=r^{2}+o_{1}(r^{\frac{5}{2}})\text{ }\text{ }\text{ as }\text{ }\text{ }r\rightarrow 0\text{ }\text{ }\text{ in asymptotically flat}\text{ }\text{ }M_{end}^{-},
\end{equation}
\begin{equation}\label{54}
u=r+o_{1}(r^{\frac{3}{2}})\text{ }\text{ }\text{ as }\text{ }\text{ }r\rightarrow 0\text{ }\text{ }\text{ in asymptotically cylindrical}\text{ }\text{ }M_{end}^{-}.
\end{equation}

In order to facilitate the construction of sub and supersolutions for  equation \eqref{48.1}, asymptotics for $f$ will be imposed which are more detailed than those in \eqref{11}-\eqref{13}. In particular
\begin{equation}\label{013}
f(r, \theta, \phi)=\sqrt{1+r^{2}} + \mathcal{A}\log{r} + \mathcal{B}(\theta, \phi)+\frac{\mathcal{D}_{1}(\theta,\phi)}{r} + \frac{\mathcal{D}_{2}(\theta, \phi)}{r^{2}} + \widetilde{f}(r, \theta, \phi)\text{ }\text{ as }\text{ }r\rightarrow\infty\text{ }\text{ in }\text{ }M_{end}^{+},
\end{equation}
where $\mathcal{A}$, $\mathcal{B}$ are given in \eqref{12} and $\mathcal{D}_{1}$, $\mathcal{D}_{2}$ functions on $S^{2}$ such that
\begin{align} \label{014-1}
\begin{split}
\mathcal{D}_{1} (\theta, \phi) =& \mathcal{A} \mathcal{C}_{1} - \mathcal{C}_{1}^{2} + 2 \mathcal{C}_{2}  \\ %= -3\overline{m}^{2} + 2C_{2}%
6\mathcal{D}_{2} (\theta, \phi) =& b^{(5)}_{rr}-2\mathbf{m}^{r} + \frac{1}{8\pi}\int_{S^{2}}\left[Tr_{\sigma}( \mathbf{m}^{g} + 2\mathbf{m}^{k})+2\mathbf{m}^{r}\right] \\
&+ 6\mathcal{C}_{3}-2\mathcal{C}_{1}\mathcal{D}_{1} + 4\mathcal{A} \mathcal{C}_{2} -2\mathcal{A} \mathcal{C}_{1}^{2} -6\mathcal{C}_{1}\mathcal{C}_{2} +2\mathcal{C}_{1}^{3}
%=\frac{1}{6} \left( \textbf{b} + \frac{3Tr_{\sigma} \textbf{m}}{2} + 6C_{3} +18 \overline{m} C_{2} -12 \overline{m}^{3} \right)%
\end{split}
\end{align}
and
\begin{equation} \label{014}
\widetilde{f}=O_{2}(r^{-3}).
\end{equation}
On the other end the following asymptotics will be imposed
\begin{equation}\label{50}
r^{-1}|\nabla f|_{g}+r^{-2}|\nabla^{2} f|_{g}\leq C\text{ }\text{ }\text{ in asymptotically flat }\text{ }\text{ }M_{end}^{-},
\end{equation}
\begin{equation}\label{51}
|\nabla f|_{g}+|\nabla^{2} f|_{g}\leq Cr^{\frac{1}{2}}\text{ }\text{ }\text{ in asymptotically cylindrical }\text{ }\text{ }M_{end}^{-}.
\end{equation}

\begin{lemma}\label{lemma31}
If $M_{end}^{+}$ is asymptotically hyperboloidal and \eqref{43}, \eqref{52}, \eqref{013}, and \eqref{014} are satisfied then the data set $(\overline{M}_{end}^{\!\text{ }+},\overline{g},\overline{k})$ is asymptotically flat. Furthermore, the mass is given by $\overline{m}_{adm}=2m+\mathcal{C}_{1}$.
\end{lemma}

\begin{proof}
Asymptotic flatness follows from Lemma \ref{lemma1} and the fact that $Y^{\phi}=O(r^{-3})$. More precisely, the computations from Lemma 2.1 of \cite{ChaKhuri1} yield
\begin{equation}\label{0025}
Y_{\phi}=g_{\phi\phi}Y^{\phi},\text{ }\text{ }\text{ }\text{ } Y_{i}=\overline{g}_{ij}Y^{j}=\overline{g}_{i\phi}Y^{\phi}
=(g_{i\phi}+f_{i}Y_{\phi})Y^{\phi}
=(g_{i\phi}+f_{i}g_{\phi\phi}Y^{\phi})Y^{\phi},
\end{equation}
and
\begin{equation}\label{0026}
\overline{g}_{ij}=g_{ij}+(f_{i}g_{j\phi}+f_{j}g_{i\phi})Y^{\phi}
+(u^{2}+g_{\phi\phi}(Y^{\phi})^{2})f_{i}f_{j},
\end{equation}
so that
\begin{equation}\label{0001}
\overline{g}_{ij}
=g_{ij}+u^{2}f_{i}f_{j}+O(r^{-2})
=\delta_{ij}+O(r^{-1}),\text{ }\text{ }\text{ }\text{ }\text{ }\text{ }
\overline{k}_{ij}=O(r^{-2}).
\end{equation}
Estimates on the derivatives of $\overline{g}_{ij}$ may be obtained in a straightforward way.
%, and therefore \eqref{9} holds.
Moreover, the formula for the mass follows from the expansion \eqref{52} and a similar computation to that in the proof of Lemma \ref{lemma1}.  %Note : To be precise, \overline{m} = A + C_{1}%
\end{proof}

We now show how to choose $u$. In light of Lemma \ref{lemma31}, and the fact that the deformed
data are simply connected and axially symmetric, the results of \cite{Chrusciel} (see also \cite{Sokolowsky}) apply to yield a global Brill coordinate system $(\overline{\rho},\phi,\overline{z})$ such that
\begin{equation}\label{55}
\overline{g}=e^{-2\overline{U}+2\overline{\gamma}}
(d\overline{\rho}^{2}+d\overline{z}^{2})
+\overline{\rho}^{2}e^{-2\overline{U}}
(d\phi+A_{\overline{\rho}}d\overline{\rho}
+A_{\overline{z}}d\overline{z})^{2}.
\end{equation}
Note that for this one may need to assume some additional regularity of the initial data. The asymptotics of the asymptotically flat end $\overline{M}_{end}^{\!\text{ }+}$ are given by
\begin{equation}\label{0017}
\overline{U}=O_{1}(r^{-1}),\text{ }\text{ }\text{ }\text{ }\overline{\gamma}=o_{1}(r^{-1}),\text{ }\text{ }\text{
}\text{ }A_{\overline{\rho}}=\overline{\rho} o_{1}(r^{-\frac{5}{2}}),\text{ }\text{ }\text{ }\text{ }A_{\overline{z}}=o_{1}(r^{-\frac{3}{2}}),
\end{equation}
and the asymptotics for the other end $\overline{M}_{end}^{\!\text{ }-}$ depend on whether it is asymptotically flat or asymptotically cylindrical in the following way
\begin{equation}\label{0018}
\overline{U}=2\log r+o_{1}(r^{\frac{1}{2}}),\text{ }\text{ }\text{ }\text{ }\overline{\gamma}=o_{1}(r^{\frac{1}{2}}),\text{
}\text{ }\text{ }\text{ }A_{\overline{\rho}}=\overline{\rho} o_{1}(r^{\frac{1}{2}}),\text{ }\text{ }\text{ }\text{
}A_{\overline{z}}=o_{1}(r^{\frac{3}{2}}),
\end{equation}
\begin{equation}\label{0019}
\overline{U}=\log r+o_{1}(r^{\frac{1}{2}}),\text{ }\text{ }\text{ }\text{ }\overline{\gamma}=o_{1}(r^{\frac{1}{2}}),\text{
}\text{ }\text{ }\text{ }A_{\overline{\rho}}=\overline{\rho} o_{1}(r^{\frac{1}{2}}),\text{ }\text{ }\text{ }\text{
}A_{\overline{z}}=o_{1}(r^{\frac{3}{2}}),
\end{equation}
respectively. By following the arguments in \cite{ChaKhuri1} we find
\begin{equation}\label{56}
\overline{m}_{adm}-\mathcal{M}(\overline{U},\overline{\omega})
\geq\frac{1}{8\pi}\int_{\overline{M}}\frac{e^{\overline{U}}}{u}
\operatorname{div}_{\overline{g}}(uQ),
\end{equation}
where $\overline{\omega}$ is the so-called twist potential function and
\begin{equation}\label{57}
\mathcal{M}(\overline{U},\overline{\omega})
=\frac{1}{32\pi}\int_{\mathbb{R}^{3}}4|\partial \overline{U}|^{2}+\frac{e^{4\overline{U}}}{\overline{\rho}^{4}}
|\partial\overline{\omega}|^{2}.
\end{equation}
This suggests that we choose
\begin{equation}\label{58}
u=e^{\overline{U}}.
\end{equation}
Moreover, if $u$ has an expansion of the form \eqref{52} then it must hold that $\mathcal{C}_{1} = -\overline{m}_{adm}$ as desired, since with the asymptotics \eqref{0017} the mass is given by (see \cite{Chrusciel})
\begin{equation} \label{58-2}
\overline{m}_{adm} = \frac{1}{4\pi}  \int_{S_{\infty}} \partial_{r}\overline{U}.
\end{equation}

\begin{theorem}\label{thm7}
Let $(M,g,k)$ be a smooth, simply connected, axially symmetric initial data set satisfying the dominant energy condition $\mu\geq|J|$ and condition \eqref{34}, and with two ends, one designated asymptotically hyperboloidal and the other either asymptotically flat or asymptotically cylindrical. If the system of equations \eqref{42}, \eqref{48.1}, \eqref{58} admits a smooth solution $(u,Y^{\phi},f)$ satisfying the asymptotics described above, then
\begin{equation}\label{59}
m\geq\sqrt{|\mathcal{J}|}
\end{equation}
and if equality is achieved then the initial data arise from an embedding into the extreme Kerr spacetime.
\end{theorem}

\begin{proof}
In light of \eqref{58}, an application of the divergence theorem yields boundary terms on the right-hand side of \eqref{56}. It is shown in \cite{ChaKhuri1} that the inner boundary integral vanishes as a consequence of $\overline{\mathcal{J}}=\mathcal{J}$, and in Appendix B it is shown that the outer (at spatial infinity) boundary integral vanishes. Therefore  $\overline{m}_{adm}\geq\mathcal{M}(\overline{U},\overline{\omega})$.
Furthermore, Dain \cite{Dain0} as well as Schoen and Zhou \cite{SchoenZhou} have shown that
$\mathcal{M}(\overline{U},\overline{\omega})\geq\sqrt{|\overline{\mathcal{J}|}}$. Thus, according to \eqref{36} the desired inequality holds.
Also, the case of equality may be treated directly from the arguments in \cite{ChaKhuri1}.
\end{proof}

In order to lend further credence to the above procedure we show that solutions to equation \eqref{48.1}
exist with the desired asymptotics. Note that the equation \eqref{42}, for $Y^{\phi}$, has already been shown to be uniquely solvable with the desired asymptotics \eqref{43} in \cite{ChaKhuri1}, given $u$ and $f$. %As with Theorem \ref{thm2}, we assume here that $(g,k)$ takes the form %\eqref{1} with $a$ and $b$ as in \eqref{2}-\eqref{3} additionally satisfying $a_{rr} = %a_{r\alpha} = 0$ and $b_{rr} = O(r^{-5})$; in this case $\mathbf{m}^{r}=0$.

\begin{theorem}\label{thm6}
Given a smooth positive function $u$ satisfying \eqref{52}-\eqref{54}, and smooth function $Y^{\phi}$ satisfying \eqref{43}-\eqref{00134}, there exists a smooth  solution $f$ to equation \eqref{48.1} satisfying \eqref{013}, \eqref{014-1}, \eqref{50}, and \eqref{51}.
\end{theorem}

\begin{proof}
This result may be proven in the same fashion as the analogous result  \cite{ChaKhuri1} in the asymptotically flat case, once suitable sub and supersolutions of \eqref{48.1} are constructed.

Let $f_0$ be a smooth function defined on $M$ such that
\begin{equation} \label{017}
f_{0}=\sqrt{1+r^{2}} + \mathcal{A}\log{r} + \mathcal{B}(\theta, \phi) + \frac{\mathcal{D}_{1}(\theta,\phi)}{r} + \frac{\mathcal{D}_{2}(\theta, \phi)}{r^{2}}  \text{ }\text{ }\text{ }\text{ in }\text{ }\text{ } \text{ }M_{end}^{+},
\end{equation}
where $\mathcal{A}$, $\mathcal{B}$, $\mathcal{D}_{1}$, and $\mathcal{D}_{2}$ are given in \eqref{12} and \eqref{014-1}, and such that $f_{0}=0$ on $M_{end}^{-}$. If $f$ is a solution of \eqref{48.1} and $h=f-f_0$ then
\begin{align}\label{48.3}
\begin{split}
 & \Delta_{g} h  + 2 \left \langle \frac{\nabla u}{u}, \nabla h \right\rangle - (Tr_g k)\left(\sqrt{u^{-2} + |\nabla f_0 + \nabla h|^2 _g} - \sqrt{u^{-2} + |\nabla f_0|^2 _g}\right) \\ &  \quad  + \left( \Delta_{g} f_0  + 2 \left \langle \frac{\nabla u}{u}, \nabla f_0 \right \rangle - (Tr_g k) \sqrt{u^{-2} + |\nabla f_0 |^2 _g} \right) = 0.
\end{split}
\end{align}

The expressions for the coefficients of \eqref{48.3} in $M_{end}^-$ can be found in \cite{ChaKhuri1}. As for $M_{end}^+$, with the help of
\begin{align}
\begin{split}
&g^{rr} = (1+r^{2}) - \frac{\mathbf{m}^{r}}{r} + O(r^{-2}), \text{ }\text{ }\text{ }\text{ }\text{ }\text{ }\text{ }\text{ }\text{ }\text{ }
g^{r \alpha} = - a_{r \alpha} \sigma^{\alpha \alpha} \left(1 + \frac{1}{r^{2}}\right) + O(r^{-6} ), \text{ }\text{ }\text{ } \\
&g^{\alpha \alpha} = \sigma^{\alpha \alpha}\left(\frac{1}{r^{2}} - \frac{\sigma^{\alpha\alpha}\mathbf{m}^{g}_{\alpha\alpha}}{r^{5}} \right)+ O(r^{-6}), \text{ }\text{ }\text{ }\text{ }\text{ }\text{ }\text{ } g^{\alpha\beta} = O(r^{-5}),\text{ }\text{ }\text{ }\text{ }\text{ }\text{ }\alpha\neq\beta,
\end{split}
\end{align}
and
\begin{equation}
\Gamma^{\alpha}_{rr} = O(r^{-5}),\text{ }\text{ } \text{ }\text{ }\text{ }\text{ }
\Gamma^{r}_{r \alpha} = O(r^{-2}), \text{ }\text{ }\text{ }\text{ }\text{ }\text{ }
\Gamma^{\beta}_{r \alpha} = O(r^{-1}),\text{ }\text{ }\text{ }\text{ }\text{ }\text{ }
\Gamma^{\gamma}_{\alpha \beta} = (\Gamma_{\sigma})^{\gamma}_{\alpha \beta} + O(r^{-3}),
\end{equation}
the asymptotics for the terms in the second line of \eqref{48.3} may be computed as follows
\begin{align} \label{018-1}
\begin{split}
\Delta_{g} f_{0}
=& g^{rr}\left(\partial_{r}^{2}f_{0}-\Gamma^{r}_{rr} \partial_{r}f_{0}\right) - g^{\alpha \beta}\Gamma_{\alpha \beta}^{r}\partial_{r}f_{0} + \frac{\Delta_{\sigma}\mathcal{B}}{r^{2}} + O(r^{-3}) \\
=& g^{rr}\partial_{r}^{2}f_{0} + \frac{g^{rr}\partial_{r}f_{0}}{2}\left( \frac{2r}{1+r^{2}} + \frac{4}{r} - \frac{3Tr_{\sigma}\textbf{m}^{g}}{r^{4}}+\frac{3\mathbf{m}^{r}}{r^{4}} \right) \\
&+ \frac{1}{2r^{2}}\left(\left[Tr_{\sigma}\left( \mathbf{m}^{g} + 2\mathbf{m}^{k} \right)+2\mathbf{m}^{r}\right]-\frac{1}{4\pi}\int_{S^{2}}\left[Tr_{\sigma}\left( \mathbf{m}^{g} + 2\mathbf{m}^{k} \right)+2\mathbf{m}^{r}\right] \right) + O(r^{-3}) \\
=& 3 \sqrt{1+r^{2}} + 2\mathcal{A} -\frac{\mathcal{D}_{1}}{r} + \frac{\mathcal{A}}{r^{2}}
-\frac{Tr_{\sigma}\textbf{m}^{g}}{r^{2}} + \frac{Tr_{\sigma}\mathbf{m}^{k}}{r^{2}}
-\frac{\mathbf{m}^{r}}{2r^{2}}\\
&
-\frac{1}{8\pi r^{2}}\int_{S^{2}}\left[Tr_{\sigma}\left( \mathbf{m}^{g} + 2\mathbf{m}^{k} \right)
+2\mathbf{m}^{r}\right]+ O(r^{-3}),
\end{split}
\end{align}
\begin{align} \label{018-2}
\begin{split}
2 \left\langle \frac{\nabla u}{u}, \nabla f_{0}\right \rangle
=& \frac{2g^{rr}}{u}\partial_{r}u\partial_{r}f_{0} + O(r^{-3})\\
=& -2\mathcal{C}_{1} + \frac{2\mathcal{C}_{1}^{2}-4\mathcal{C}_{2}-2\mathcal{A}\mathcal{C}_{1}}{r} \\
&+ \frac{2\mathcal{C}_{1}\mathcal{D}_{1}-\mathcal{C}_{1}
+6\mathcal{C}_{1}\mathcal{C}_{2}-6\mathcal{C}_{3}
-4\mathcal{A}\mathcal{C}_{2}+2\mathcal{A}\mathcal{C}_{1}^{2}
-2\mathcal{C}_{1}^{3}}{r^{2}}  + O(r^{-3}).
%&= 2 \overline{m} + \frac{6 \overline{m}^{2}-4C_{2}}{r} + \frac{\overline{m}-18\overline{m}C_{2}-6C_{3}+12\overline{m}^{3}}{r^{2}} + O(r^{-3})%
\end{split}
\end{align}
Moreover
\begin{align} \label{0018-3}
\begin{split}
\sqrt{u^{-2}+ |\nabla f_{0}|^{2}}
=&
\left( g^{rr}(\partial_{r}f_{0})^{2}+ \left( 1 - \frac{2\mathcal{C}_{1}}{r}\right)  + O(r^{-2}) \right)^{1/2} \\
=&
\left(r + \mathcal{A} + \frac{1-2\mathcal{D}_{1}}{2r} -\frac{2\mathcal{C}_{1}+4\mathcal{D}_{2}+\mathbf{m}^{r}}{2r^{2}} + O(r^{-3}) \right),
\end{split}
\end{align}
so that
\begin{align} \label{018-3}
\begin{split}
&(Tr_{g}k)\sqrt{u^{-2}+ |\nabla f_{0}|^{2}} \\
=& \left(3 + \frac{b_{rr}^{(5)}+Tr_{\sigma}\left(\mathbf{m}^{k} -\textbf{m}^{g}\right)-\mathbf{m}^{r}}{r^{3}}+ O(r^{-4}) \right)\\
&\cdot
\left(r + \mathcal{A} + \frac{1-2\mathcal{D}_{1}}{2r} -\frac{2\mathcal{C}_{1}+4\mathcal{D}_{2}+\mathbf{m}^{r}}{2r^{2}} + O(r^{-3}) \right) \\
=& 3r + 3\mathcal{A} + \frac{3-6\mathcal{D}_{1}}{2r} + \frac{b_{rr}^{(5)}+Tr_{\sigma}\left(\mathbf{m}^{k} -\textbf{m}^{g}\right)-\tfrac{5}{2}\mathbf{m}^{r}-3\mathcal{C}_{1}-6\mathcal{D}_{2}}{r^{2}} + O(r^{-3}).
\end{split}
\end{align}
Therefore, by definition of $\mathcal{A}$, $\mathcal{C}_{1}$,  $\mathcal{D}_{1}$, $\mathcal{D}_{2}$ it follows that
\begin{equation} \label{019}
\Delta_{g} f_{0} + 2 \left\langle \frac{\nabla u}{u}, \nabla f_{0} \right\rangle - (Tr_{g}k)\sqrt{u^{-2}+ |\nabla f_{0}|_{g}^{2}}  = O(r^{-3}).
\end{equation}
Observe also that
\begin{align}
 \begin{split}
 \left|\sqrt{u^{-2} + |\nabla f_0 + \nabla h|^2 _g} - \sqrt{u^{-2} + |\nabla f_0|^2 _g} \right|
 & = \frac{\left| |\nabla f_0 + \nabla h|^2 _g - |\nabla f_0|^2 _g\right|}{\sqrt{u^{-2} + |\nabla f_0 + \nabla h|^2 _g} + \sqrt{u^{-2} + |\nabla f_0|^2 _g}} \\
 & \leq \frac{\left| \left \langle \nabla h, 2 \nabla f_0 + \nabla h \right \rangle \right|}{|\nabla f_0 + \nabla h| _g + |\nabla f_0| _g} \\
 & \leq \frac{|2 \nabla f_0 + \nabla h| _g |\nabla h|_g}{|2 \nabla f_0 + \nabla h| _g} \\
 & = |\nabla h|_g.
 \end{split}
\end{align}

In order to construct radial sub and supersolutions of \eqref{48.3}, we first note that if $h_{0} = h_{0}(r)$, then
\begin{align}
 \begin{split}
  &\Delta_{g} h_{0}  + 2 \left \langle \frac{\nabla u}{u}, \nabla h_{0} \right \rangle\\
  = &g^{rr}h_{0}'' - g^{rr} \left( \Gamma^{r}_{rr} + 2(g^{rr})^{-1}g^{r \alpha}\Gamma^{r}_{r \alpha} + (g^{rr})^{-1} g^{\alpha \beta}\Gamma^{r}_{\alpha \beta} -  \frac{2\partial_{r} u}{u} -\frac{2g^{r\alpha}\partial_{\alpha}u}{ug^{rr}}\right) h_{0}' \\
=& \frac{g^{rr}}{\sqrt{1+r^2}} \left[ \zeta' -
\left( \frac{r}{1+r^{2}}+ \Gamma^{r}_{rr} + 2(g^{rr})^{-1}g^{r \alpha}\Gamma^{r}_{r \alpha} + (g^{rr})^{-1} g^{\alpha \beta}\Gamma^{r}_{\alpha \beta} -  \frac{2\partial_{r} u}{u} -\frac{2g^{r\alpha}\partial_{\alpha}u}{ug^{rr}}\right)\zeta \right],
 \end{split}
\end{align}
where $\zeta = \sqrt{1 + r^{2}} h_{0}'$.
Further let $\varTheta(r)$, $\varLambda(r)>0$, $\varUpsilon(r)>0$ be bounded radial functions such that
\begin{equation}
 \frac{\sqrt{1+r^{2}}\left|\Delta_{g} f_0  + 2 \left \langle \frac{\nabla u}{u}, \nabla f_0 \right \rangle - (Tr_g k) \sqrt{u^{-2} + |\nabla f_0 |^2 _g}\right|}{g^{rr}}\leq \varUpsilon,
\end{equation}
and
\begin{equation}
\left|\left( \frac{r}{1+r^{2}}+ \Gamma^{r}_{rr} + 2(g^{rr})^{-1}g^{r \alpha}\Gamma^{r}_{r \alpha} + (g^{rr})^{-1} g^{\alpha \beta}\Gamma^{r}_{\alpha \beta} -  \frac{2\partial_{r} u}{u} -\frac{2g^{r\alpha}\partial_{\alpha}u}{ug^{rr}}\right)+ \frac{|Tr_g k|}{\sqrt{g^{rr}}} + \varTheta \right| \leq \varLambda,
\end{equation}
where
\begin{equation}
\varTheta (r)= -\frac{1}{r}+\frac{2\overline{m}_{adm}}{r^2} + O(r^{-3}),\text{ }\text{ }\text{ }\text{ }\text{ }\text{ }\varLambda(r) = O(r^{-3}), \text{ }\text{ }\text{ }\text{ }\text{ }\text{ }\varUpsilon(r)=O(r^{-4})\text{ }\text{ }\text{ }\text{ in }\text{ }\text{ }\text{ }M_{end}^{+}.
\end{equation}
Similar to \cite{ChaKhuri1}, we (radially) extend these functions to $M \setminus M^+_{end}$ in an appropriate way and define a supersolution $h_+$ of \eqref{48.3} by
\begin{equation}
 h_+(r) =  - \int_r ^\infty \frac{\zeta_+(s)}{\sqrt{1+s^2}} ds,
\end{equation}
where
\begin{equation}
\zeta_+(r)= - e^{-\int_0^r (\varTheta (s) - \varLambda(s))ds} \int_0^r \varUpsilon(s) e^{\int_0^s (\varTheta (t) - \varLambda(t))dt} ds \leq 0
\end{equation}
is the solution of the ordinary differential equation
\begin{equation}\label{54.1}
\zeta_+' + (\varTheta - \varLambda) \zeta_+ + \varUpsilon = 0.
\end{equation}
It is straightforward to check that $h_+$ has the desired asymptotics. Furthermore,
%and define $\zeta_+ \leq 0$ to be the radial solution of the ordinary differential equation
%\begin{equation}\label{54.1}
%\zeta_+' + (\varTheta - \varLambda) \zeta_+ + \varUpsilon = 0
%\end{equation}
%such that $\zeta_+(0)=0$, that is \textcolor{red}{[[Anna: corrected typo in the formula]]}
%\begin{equation}
%\zeta_+(r)= - e^{-\int_0^r (\varTheta (s) - \varLambda(s))ds} \int_0^r \varUpsilon(s) e^{\int_0^s (\varTheta (t) - \varLambda(t))dt} ds.
%\end{equation}
%Since $\zeta_+$ is bounded, $\varUpsilon(r)=O(r^{-4})$, and $\varTheta(r) - \varLambda(r)=-\frac{1}{r} + O(r^{-2})$, it is straightforward to check using \eqref{54.1} that $\zeta_+=O(r^{-3})$ in $M_{end}^{+}$.
%We are now in a position to define a supersolution $h_+$ of \eqref{48.3} by
%\begin{equation}
%h_+(r) =  - \int_r ^\infty \frac{\zeta_+(s)}{\sqrt{1+s^2}} ds.
%\end{equation}
since $\zeta_+ \leq 0$ it follows that
\begin{align}
\begin{split}
   &\Delta_{g} h_+  + 2 \left \langle \frac{\nabla u}{u}, \nabla h_+ \right\rangle - (Tr_g k)\left(\sqrt{u^{-2} + |\nabla f_0 + \nabla h_+|^2 _g} - \sqrt{u^{-2} + |\nabla f_0|^2 _g}\right) \\ &  \quad  + \left( \Delta_{g} f_0  + 2 \left \langle \frac{\nabla u}{u}, f_0 \right \rangle - (Tr_g k) \sqrt{u^{-2} + |\nabla f_0 |^2 _g} \right)\\
   \leq & \Delta_{g} h_+  + 2 \left \langle \frac{\nabla u}{u}, \nabla h_+ \right\rangle + |Tr_g k||\nabla h_{+}|_g + \frac{g^{rr}\varUpsilon}{\sqrt{1+r^{2}}}\\
   =& \frac{g^{rr}}{\sqrt{1+r^2}} \left[ \zeta_+' -
\left( \frac{r}{1+r^{2}}+ \Gamma^{r}_{rr} + 2(g^{rr})^{-1}g^{r \alpha}\Gamma^{r}_{r \alpha}\right)\zeta_+ \right]\\
& -\frac{g^{rr}}{\sqrt{1+r^2}} \left[\left( (g^{rr})^{-1} g^{\alpha \beta}\Gamma^{r}_{\alpha \beta} -  \frac{2\partial_{r} u}{u}
  -\frac{2g^{r\alpha}\partial_{\alpha}u}{ug^{rr}}+ \frac{|Tr_g k|}{\sqrt{g^{rr}}} \right)\zeta_+ - \varUpsilon\right] \\
   \leq & \frac{g^{rr}}{\sqrt{1+r^2}}  (\zeta_+ ' + (\varTheta - \varLambda) \zeta_+ + \varUpsilon) \\
   =&0.
\end{split}
\end{align}
Consequently, $f_+ = f_0 + h_+$ is a supersolution of \eqref{48.1} which possesses the expansion \eqref{013}. Finally, since $\zeta_- = - \zeta_+$ is a solution of the ordinary differential equation
\begin{equation}
\eta_-' + (\varTheta - \varLambda) \eta_- - \varUpsilon = 0,
\end{equation}
one may similarly check that $f_- = f_0 - h_+ $ is a subsolution of \eqref{48.1}, with the desired asymptotic behavior in $M_{end}^{+}$, and such that $f_- \leq f_+$. %Note that although $h_{\pm}$ are sub and supersolutions of \eqref{48.3} on $M_{end}^{+}$, they may be extended (radially) to have this property globally as in \cite{ChaKhuri1}. The functions $f_{\pm}$ are then globally defined sub and supersolutions of \eqref{48.1} from which the existence result may be derived.
\end{proof}

\begin{remark}
It should be the case that the solutions produced in Theorem \ref{thm6} also satisfy \eqref{014}, which is important when applying the results of \cite{Chrusciel} to obtain Brill coordinates for the deformed data. However, the estimates needed for \eqref{014} are not clearly derived from the typical rescaling argument found in \cite{ChaKhuri1}. Rather, it is likely that these estimates may be derived from the techniques in \cite{Sakovich} (see Remark \ref{remark001}).
\end{remark}

\section{The Mass-Angular Momentum-Charge Inequality}
\label{sec7} \setcounter{equation}{0}
\setcounter{section}{7}

Let $(M, g, k, E, B)$ be an axisymmetric initial data set for the Einstein-Maxwell equations, with the same assumptions on $(M,g,k)$ as in Section \ref{sec6}, except that \eqref{34} is replaced by
\begin{equation}\label{7.1}
J_{EM}^{i}\eta_{i}=0.
\end{equation}
Axisymmetry entails that $\mathfrak{L}_{\eta} E=\mathfrak{L}_{\eta} B=0$, and in addition it will be required that $E$, $B$ are divergence free.
%\begin{equation} \label{60}
%\end{equation}
The asymptotics for $(E,B)$ in the asymptotically hyperboloidal end $M_{end}^{+}$ are given in \eqref{6}, whereas the asymptotics in $M_{end}^{-}$ are given separately in the asymptotically flat and asymptotically cylindrical cases by
\begin{equation}\label{618.1}
E_{i}=O_{1}(1),\text{ }\text{ }\text{ }\text{ }E_{\phi}=O_{1}(r),\text{ }\text{ }\text{ }\text{ }
B_{i}=O_{1}(1),\text{ }\text{ }\text{ }\text{ }B_{\phi}=O_{1}(r),
\text{ }\text{ }\text{ }i=\rho,z,
\end{equation}
and
\begin{equation}\label{618.2}
E_{i}=O_{1}(r^{-1}),\text{ }\text{ }\text{ }\text{ }E_{\phi}=O_{l}(1),\text{ }\text{ }\text{ }\text{ }
B_{i}=O_{1}(r^{-1}),\text{ }\text{ }\text{ }\text{ }B_{\phi}=O_{1}(1),
\text{ }\text{ }\text{ }i=\rho,z,
\end{equation}
where $(\rho,\phi,z)$ are Brill coordinates in this end with $r=\sqrt{\rho^{2}+z^{2}}$. In this setting the angular momentum is defined to be
\begin{equation}\label{7.2}
\mathcal{J}=\frac{1}{8\pi}\int_{S}(k_{ij}-(Tr_{g} k)g_{ij})\nu^{i}_{g}\eta^{j}-\frac{1}{4\pi}\int_{S}\psi_{B} E_{i}\nu_{g}^{i},
\end{equation}
where $S$ is any surface enclosing the origin with unit outer normal $\nu_{g}$, and $\psi_{B}$ is the potential for the magnetic field (see \cite{DainKhuriWeinsteinYamada} and \cite{KhuriWeinstein1}). The condition \eqref{7.1} ensures that this is well-defined \cite{DainKhuriWeinsteinYamada}. Moreover since $\psi_{B}=O(r^{-1})$
as $r\rightarrow\infty$ we have that
\begin{equation}
\int_{S_{\infty}}\psi_{B} E_{i}\nu_{g}^{i}=0,
\end{equation}
and hence \eqref{7.2} agrees with the limit definition in \eqref{33.1}.
The mass-angular momentum-charge inequality states that
\begin{equation}\label{615}
m^2\geq\frac{\mathcal{Q}^{2}+\sqrt{\mathcal{Q}^{4}+4\mathcal{J}^{2}}}{2}.
\end{equation}

We seek a deformation of the initial data $(M,g,k,E,B)\rightarrow(\overline{M},\overline{g},\overline{k},\overline{E},\overline{B})$ such that $M\cong\overline{M}$, $\overline{M}_{end}^{+}$ is asymptotically flat, and
\begin{equation}\label{621}
\overline{m}_{adm}=m,\text{ }\text{ }\text{ }\text{ }\text{ }\overline{\mathcal{J}}=\mathcal{J},\text{ }\text{ }\text{ }\text{ }\text{ }Tr_{\overline{g}}\overline{k}=0,\text{ }\text{ }\text{ }\text{ }\text{ }
\overline{R}\geq|\overline{k}|_{\overline{g}}^{2}
+2(|\overline{E}|_{\overline{g}}^{2}+|\overline{B}|^{2}_{\overline{g}})\text{ }\text{ weakly,}
\end{equation}
\begin{equation}\label{622}
\operatorname{div}_{\overline{g}}\overline{E}
=\operatorname{div}_{\overline{g}}\overline{B}=0,\text{ }\text{ }\text{ }\text{ }\text{ }\text{ }\overline{J}_{EM}(\eta)=0,\text{ }\text{ }\text{ }\text{ }\text{ }\text{ }\overline{\mathcal{Q}}_{e}=\mathcal{Q}_{e},\text{ }\text{ }\text{ }\text{ }\text{ }\text{ }\overline{\mathcal{Q}}_{b}=\mathcal{Q}_{b},
\end{equation}
where $\overline{J}_{EM}$ is the momentum density minus the electromagnetic contribution of the new data. The structure of the deformation will be the same as that in the previous section. In particular, $\overline{g}$ and $\overline{k}$ are given by \eqref{39} and $Y$ satisfies \eqref{41}. It follows that the new data are again determined by three functions $(u,Y^{\phi},f)$, and $Tr_{\overline{g}}\overline{k}=0$. The functions $u$ and $f$ are chosen according to \eqref{58} and \eqref{48.1}, with the asymptotics \eqref{52}-\eqref{54} and \eqref{013}-\eqref{51}. The function $Y^{\phi}$ is chosen here to satisfy a slightly different equation, namely
\begin{equation}\label{630}
div_{\overline{g}}\overline{k}(\eta)+2\overline{E}\times\overline{B}(\eta)=0,
\end{equation}
but will keep the same asymptotics \eqref{43}-\eqref{00134}. This equation is equivalent to $\overline{J}_{EM}(\eta)=0$, and guarantees the existence of a charged twist potential \cite{ChaKhuri2} in addition to a well-defined angular momentum by \eqref{7.2}. As in the previous section, we then have $\overline{m}_{adm}=m$ and $\overline{\mathcal{J}}=\mathcal{J}$.

The deformation of the electromagnetic field will follow the construction in \cite{ChaKhuri2}. Let $(e_{1}, e_{2}, e_{3}=|\eta|^{-1}\eta)$ be an orthonormal frame for $(M,g)$, and set
%\begin{equation}\label{645}
%\widetilde{E}(e_{i})=E(e_{i}), \text{ }\text{ }\text{ }\text{ }\text{ %}\text{ }\text{ }  \widetilde{B}(e_{i})= B(e_{i}) \text{ }\text{ }\text{ %}\text{ }\text{ }\text{ for }\text{ }\text{ }\text{ }\text{ }\text{ } i %= 1,2,
%\end{equation}
%\begin{equation}\label{646}
%\widetilde{E}(e_{3})=v \times B(e_{3}),\text{ }\text{ }\text{ }\text{ %}\text{ }\text{ }\text{ }\text{ }
%\widetilde{B}(e_{3})=-v \times E(e_{3}),
%\end{equation}
%where $v = \frac{u \nabla f}{\sqrt{1+ u^{2}|\nabla f|_{g}^{2}}}$. And %then, we choose $(\overline{E}, \overline{B})$ as
\begin{equation}\label{665}
\overline{E}(e_{i})=\frac{E(e_{i})}{\sqrt{\volg}},\text{ }\text{ }\text{ }\text{ }\text{ }\overline{B}(e_{i})=\frac{B(e_{i})}{\sqrt{\volg}}\text{ }\text{ }\text{ }\text{ for }\text{ }\text{ }\text{ }i=1,2,\text{ }\text{ }\text{ }\text{ }\text{ }\overline{E}(e_{3})=\overline{B}(e_{3})=0.
\end{equation}
%\begin{equation}\label{665.1}
%\overline{E}(e_{3})=\overline{B}(e_{3})=0.
%\end{equation}
Then it follows from \cite{ChaKhuri2} that
\begin{align}\label{694}
\begin{split}
\overline{R}-|\overline{k}|_{\overline{g}}^{2}-2(|\overline{E}|_{\overline{g}}^{2}+|\overline{B}|_{\overline{g}}^{2})
=& 2(\mu_{EM}-J_{EM}(v))+|k-\pi|_{g}^{2}+2u^{-1}\operatorname{div}_{\overline{g}}(uQ)\\
&+2\left(E(e_{3})- v\times B(e_{3})\right)^{2}  + 2\left(B(e_{3})+ v\times E(e_{3})\right)^{2},
\end{split}
\end{align}
and
\begin{equation}\label{695}
\operatorname{div}_{\bg}\be
= \frac{\operatorname{div}_{g} E}{\sqrt{\volg}}=0, \text{ }\text{ }\text{ }\text{ }\text{ }\text{ }\text{ }\text{ }
\operatorname{div}_{\bg}\bm = \frac{\operatorname{div}_{g} B}{\sqrt{\volg}}=0.
\end{equation}
Moreover a similar computation as in \eqref{23} shows that
\begin{equation} \label{696}
\overline{E}_{i}\nu_{\overline{g}}^{i} =E_{i}\nu_{g}^{i}+O(r^{-3}),
\end{equation}
which ensures that $\overline{\mathcal{Q}}_{e}=\mathcal{Q}_{e}$ and
$\overline{\mathcal{Q}}_{b}=\mathcal{Q}_{b}$.

\begin{theorem}\label{thm8}
Let $(M,g,k,E,B)$ be a smooth, simply connected, axially symmetric initial data set satisfying the charged dominant energy condition $\mu_{EM}\geq|J_{EM}|$ and condition \eqref{7.1}, and with two ends, one designated asymptotically hyperboloidal and the other either asymptotically flat or asymptotically cylindrical. If the system of equations \eqref{630}, \eqref{48.1}, \eqref{58} admits a smooth solution $(u,Y^{\phi},f)$ satisfying the asymptotics described above, then
\begin{equation}\label{59}
m^2\geq\frac{\mathcal{Q}^{2}+\sqrt{\mathcal{Q}^{4}+4\mathcal{J}^{2}}}{2},
\end{equation}
and if equality is achieved then the initial data arise from an embedding into the extreme Kerr-Newman spacetime.
\end{theorem}

\begin{proof}
As in the proof of Theorem \ref{thm7}, we can apply the arguments from the asymptotically flat case \cite{ChaKhuri2} with only minor modifications, in light of \eqref{621} and \eqref{622}.
The only difference arises from the boundary integral at null infinity, which is shown to vanish in  Appendix B.
\end{proof}

\begin{remark}
Note that the system of equations associated with Theorem \ref{thm8} is exactly the same as that of Theorem \ref{thm7}, save for a minor (lower order) modification in the equation for $Y^{\phi}$ \eqref{630}. Thus, each equation may be solved independently with the appropriate asymptotics, lending further support to the above procedure.
\end{remark}

\section{Lower Bounds for Area in Terms of Mass, Angular Momentum, and Charge}
\label{sec8} \setcounter{equation}{0}
\setcounter{section}{8}

Here we point out another application of the reduction procedure presented in the previous section. Let $(M,g,k,E,B)$ be as in Theorem \ref{thm8}, then heuristic physical arguments \cite{DainKhuriWeinsteinYamada} lead to the inequality
\begin{equation}\label{7153}
\frac{A_{min}}{8\pi}\geq m^2-\frac{\mathcal{Q}^2}{2}
-\sqrt{\left(m^2-\frac{\mathcal{Q}^2}{2}\right)^2
-\frac{\mathcal{Q}^4}{4}-\mathcal{J}^2}
\end{equation}
where $A_{min}$ is the minimum area required to enclose $M_{end}^{-}$.
In \cite{DainKhuriWeinsteinYamada} this has been proven in the maximal case when $M_{end}^{+}$ is asymptotically flat and $m$ represents the ADM mass. The proof follows directly from the mass-angular momentum-charge inequality \cite{ChruscielCosta}, and the area-angular momentum-charge inequality \cite{ClementJaramilloReiris}. In the general (non-maximal) case, the area-angular momentum-charge inequality has been established when $A_{min}$ is replaced by the area of a stable, axisymmetric, marginally outer trapped surface \cite{ClementJaramilloReiris}. For the mass-angular momentum-charge inequality, we have shown how to reduce the case of an asymptotically hyperboloidal end to that of an asymptotically flat end, modulo the problem of solving a coupled system of elliptic equations. Therefore a lower bound for area analogous to \eqref{7153} may also be reduced to the same problem, by combining Theorem \ref{thm8} above with the proof of a Theorem 2.5 in \cite{DainKhuriWeinsteinYamada}.

\begin{theorem}\label{thm9}
Let $(M,g,k,E,B)$ be a smooth, simply connected, axially symmetric initial data set satisfying the charged dominant energy condition $\mu_{EM}\geq|J_{EM}|$ and condition \eqref{7.1}, and with two ends, one designated asymptotically hyperboloidal and the other either asymptotically flat or asymptotically cylindrical. If the data possesses a stable axially symmetric marginally outer trapped surface with area $A$, and the system of equations \eqref{630}, \eqref{48.1}, \eqref{58} admits a smooth solution $(u,Y^{\phi},f)$ satisfying the asymptotics \eqref{43}-\eqref{00134}, \eqref{52}-\eqref{54}, and \eqref{013}-\eqref{51} then
\begin{equation}\label{154}
\frac{A}{8\pi}\geq m^2-\frac{\mathcal{Q}^2}{2}
-\sqrt{\left(m^2-\frac{\mathcal{Q}^2}{2}\right)^2
-\frac{\mathcal{Q}^4}{4}-\mathcal{J}^2},
\end{equation}
and if equality is achieved then the initial data arise from an embedding into the extreme Kerr-Newman spacetime.
\end{theorem}

\section{Appendix A: Adjustment of Asymptotically Hyperboloidal Initial Data}
\label{sec9} \setcounter{equation}{0}
\setcounter{section}{9}

In this appendix we will show that given asymptotically hyperboloidal initial data $(M,g,k)$ as described in Section \ref{sec2}, one can perform a change of coordinates at infinity to achieve $a_{rr} = a_{r\alpha} = 0$ without affecting the mass aspect function. Following \cite{CortierDahlGicquaud} we refer to this change of coordinates as an `adjustment'. We note that adjustment is a standard procedure which becomes relevant when considering deformations of asymptotically hyperbolic conformally compact metrics or their evolution under geometric flows, see for example the two references mentioned below.

Let $g$ be as in Section \ref{sec2}, so that it takes the asymptotics form
\begin{equation}
 \left( \frac{1}{1+r^2} + \frac{\mathbf{m}^r}{r^5} + O(r^{-6}) \right)dr^2 + \left(\frac{2\mathbf{a}_\alpha}{r^3} + O(r^{-4})\right) dr dy^{\alpha} + \left( r^2 \sigma_{\alpha\beta} + \frac{\mathbf{m}^g_{\alpha\beta}}{r} + O(r^{-2})\right) dy^\alpha dy^\beta,
\end{equation}
where $\mathbf{a}$ is an $r$-independent 1-form on $S^2$. Using the substitution
\begin{equation}
r = \sinh^{-1} \varrho,
\end{equation}
we bring $g$ to the conformally compact form
\begin{equation}
\sinh^{-2}\! \varrho \left[ (1+ \mathbf{m}^r \varrho^3\! + O(\varrho^4)) d\varrho^2 \!+ \left(2\mathbf{a}_\alpha \varrho^3\!  + O(\varrho^4) \right) d\varrho dy^{\alpha}\! +\! \left( \sigma_{\alpha\beta} + \mathbf{m}^g_{\alpha\beta} \varrho^3\! + O(\varrho^4)\right) dy^\alpha dy^\beta \right].
\end{equation}

Now, similar to \cite[Section III]{BalehowskyWoolgar}, we apply the coordinate transformation
\begin{equation}
x^\alpha = y^\alpha + \frac{\varrho}{4} \sigma^{\alpha\beta} \mathbf{a}_\beta
\end{equation}
to obtain
\begin{equation}
\sinh^{-2} \varrho \left[ (1+ \mathbf{m}^r \varrho^3 + O(\varrho^4)) d\varrho^2 + O(\varrho^4) d\varrho dx^{\alpha} + \left( \sigma_{\alpha\beta} + \mathbf{m}^g_{\alpha\beta} \varrho^3 + O(\varrho^4)\right) dx^\alpha dx^\beta \right].
\end{equation}
A further substitution
\begin{equation}
\varrho' = \varrho + \frac{\mathbf{m}^r}{6} \varrho^3
\end{equation}
yields
\begin{equation}
\sinh^{-2}\! \varrho' \left[ (1 + O(\varrho'^4)) d\varrho^2\! + O(\varrho'^4) d\varrho' dx^{\alpha}\! + \left( \sigma_{\alpha\beta} + \left(\mathbf{m}^g_{\alpha\beta}\! + \tfrac{1}{3}\mathbf{m}^r \sigma_{\alpha\beta} \right)\varrho'^3\! + O(\varrho'^4)\right) dx^\alpha dx^\beta \right].
\end{equation}
We are now in a position to change the conformal gauge as described in \cite[Section 3.2.1]{AnderssonCaiGalloway}. This gives
\begin{equation}
g = \sinh^{-2} \tilde{\varrho} \left[ d\tilde{\varrho}^2 + \left( \sigma_{\alpha\beta} + \left(\mathbf{m}^g_{\alpha\beta} + \tfrac{1}{3}\mathbf{m}^r \sigma_{\alpha\beta} \right)\tilde{\varrho}^3 + O(\tilde{\varrho}^4)\right) dx^\alpha dx^\beta \right],
\end{equation}
for $\tilde{\varrho} = \varrho' + O(\varrho'^5)$.

To bring $g$ back to the initial form we set $\tilde{r} = \sinh^{-1} \tilde{\rho}$, thereby obtaining
\begin{equation}
g = \frac{d\tilde{r}^2}{1+\tilde{r}^2} + \left( \tilde{r}^2 \sigma_{\alpha\beta} + \frac{\mathbf{m}^g_{\alpha\beta} + \tfrac{1}{3}\mathbf{m}^r \sigma_{\alpha\beta} }{\tilde{r}} + O(\tilde{r}^{-2})\right) dx^\alpha dx^\beta.
\end{equation}
Finally, it is straightforward to check that the described change of coordinates results in
 \begin{equation}
k = \left( \frac{1}{1+\tilde{r}^2} + O(\tilde{r}^{-5}) \right)d\tilde{r}^2 + O(\tilde{r}^{-3}) d\tilde{r} dx^{\alpha} + \left( \tilde{r}^2 \sigma_{\alpha\beta} + \frac{\mathbf{m}^k_{\alpha\beta}+ \tfrac{1}{3}\mathbf{m}^r \sigma_{\alpha\beta}}{\tilde{r}} + O(\tilde{r}^{-2})\right) dx^\alpha dx^\beta,
\end{equation}
thus the mass aspect function remains unchanged.

\begin{remark}
Note that in \cite[Section 3.2.1]{AnderssonCaiGalloway} it is assumed that the metric has smooth conformal compactification. In general, the application of the conformal gauge change will result in loss of regularity of the initial data by one derivative, see e.g. \cite[Lemma 5.1]{Lee}.
\end{remark}

\section{Appendix B: Boundary Integrals}
\label{sec10} \setcounter{equation}{0}
\setcounter{section}{10}

In this section we will show that boundary integrals arising from \eqref{17} and \eqref{56} vanish.
\begin{lemma}\label{lemma3}
Under the hypotheses and notation of Theorem \ref{thm1}
\begin{equation}\label{010101}
 \lim_{\tau\rightarrow\infty}\int_{\overline{S}_{\tau}} u \overline{g}(q,\nu_{\overline{g}}) = -4\pi(2u_{0}+\mathcal{A}).
\end{equation}
In particular, this vanishes when $m=\overline{m}_{adm}$.
\end{lemma}

\begin{proof}
As explained in the proof of Theorem \ref{thm1}, the results in \cite{HuiskenIlmanen1} allow for an approximation by coordinate spheres $\overline{S}_{r}$ in the asymptotically flat end of $(\overline{M},\overline{g})$, so that
\begin{equation}
\lim_{\tau\rightarrow\infty} \int_{\overline{S}_{\tau}} u \overline{g}(q,\nu_{\overline{g}})
 =\lim_{r\rightarrow\infty}\int_{\overline{S}_{r}} u \overline{g}(q,\nu_{\overline{g}}).
\end{equation}
Let $S_{r}$ be coordinate spheres and $(e_{r},e_{\theta},e_{\phi})$ be an orthonormal frame (associated with cylindrical coordinates) in the asymptotic end of $(M,g)$. Then calculations in Appendix C of \cite{ChaKhuri1} yield
\begin{equation} \label{01}
\nu_{\overline{g}} = \sqrt{\frac{1 + u^{2}|\nabla_{S}f|^{2}}{1+u^{2}|\nabla f|_{g}^{2}}}
\left(e_{r} - \frac{u^{2}e_{r}(f)e_{\theta}(f)}{1+ u^{2}|\nabla_{S}f|^{2}} e_{\theta}
- \frac{u^{2}e_{r}(f)e_{\phi}(f)}{1+u^{2}|\nabla_{S}f|^{2}} e_{\phi} \right),
\end{equation}
where $\nabla_{S}$ denotes covariant differentiation on $S_{r}$. Since the area forms of $\overline{S}_{r}$ and $S_{r}$ differ by a factor of $\sqrt{1 + u^{2}|\nabla_{S}f|^{2}}$, it follows that
\begin{equation} \label{03}
\int_{\overline{S}_{r}} u \overline{g}(q,\nu_{\overline{g}})
= \int_{S_{r}}  \frac{u(1 + u^{2}|\nabla_{S}f|^{2})}{\sqrt{1+u^{2}|\nabla f|_{g}^{2}}}
\left(q(e_{r}) - \frac{u^{2}e_{r}(f)}{1+ u^{2}|\nabla_{S}f|^{2}} q(\nabla_{S}f)\right).
\end{equation}

We will now compute the asymptotics for each term in \eqref{03}. Recall the definition of $q$ in \eqref{15.2}, and observe that
%For $\alpha, \beta = \theta, \phi$, as $r \to \infty$, recall that
%\begin{equation} \label{04}
%k_{rr} = \frac{1}{1+r^{2}} + O(r^{-5}), \quad k_{r\alpha}= O(r^{-3}), \quad %k_{\alpha \beta}= r^{2}\sigma_{\alpha \beta} + O(r^{-1}),
%\end{equation}
%\begin{equation}\label{04-1}
%f(r, \theta, \phi)=\sqrt{1+r^{2}} + A\log{r} + + B(\theta, \phi)+\widetilde{f}%(r, \theta, \phi),
%\end{equation}
%\begin{equation}\label{06}
%\partial_{\alpha}\widetilde{f}=O(r^{-\epsilon}), \quad \partial_{r}%\widetilde{f}=O(r^{-1-\epsilon}),
%\quad \partial_{\alpha}\partial_{r}\widetilde{f}=O(r^{-\epsilon}), \quad %\partial_{r}^{2}\widetilde{f} = O(r^{-2-\epsilon}),
%\end{equation}
%and
%\begin{equation}\label{06-2}
%\phi = 1 + \frac{C}{r} + O(r^{-2}).
%\end{equation}
\begin{equation}\label{07}
\begin{split}
\pi_{rr}&= \frac{u \nabla_{rr}f + 2\partial_{r}u\partial_{r}f}{\sqrt{1+u^{2}|\nabla f|_{g}^{2}}} \\
&= \left( \frac{1}{\sqrt{1+r^{2}}} - \frac{u_{0}}{r\sqrt{1+r^{2}}} + O(r^{-3+\varepsilon})\right) \left( \frac{1}{r} - \frac{u_{0}+\mathcal{A}}{r^{2}}+ O(r^{-3+\varepsilon}) \right) \\
&=\frac{1}{r^{2}}-\frac{2u_{0}+\mathcal{A}}{r^{3}} + O(r^{-4+\varepsilon}),
\end{split}
\end{equation}
\begin{equation} \label{08}
\pi_{r \alpha} = \frac{u \nabla_{r\alpha}f + \partial_{r}u\partial_{\alpha}f + \partial_{\alpha}u\partial_{r}f}{\sqrt{1+u^{2}|\nabla f|_{g}^{2}}}
= o(1)\left( \frac{1}{r} - \frac{u_{0}+\mathcal{A}}{r^{2}}+ O(r^{-3+\varepsilon}) \right)
= o(r^{-1}),
\end{equation}
and
\begin{equation} \label{09}
\begin{split}
\pi_{\alpha \beta}=&= \frac{u\nabla_{\alpha\beta} f + \partial_{\beta}u\partial_{\alpha}f + \partial_{\alpha}u\partial_{\beta}f}{\sqrt{1+u^{2}|\nabla f|_{g}^{2}}} \\
&= \frac{1}{2}\left(u g^{rr}\partial_{r}f\partial_{r}g_{\alpha \beta} +O(1)\right)\left( \frac{1}{r} - \frac{u_{0}+\mathcal{A}}{r^{2}}+ O(r^{-3+\varepsilon}) \right)\\
&= r^{2}\sigma_{\alpha \beta} + O(r^{\varepsilon}).
\end{split}
\end{equation}
Substituting \eqref{07}, \eqref{08}, and \eqref{09} into the expression for $q$ yields
\begin{equation} \label{010}
\begin{split}
q(e_{r}) =& q(\sqrt{1+r^{2}}\partial_{r})+o(r^{-1}) \\
=& \frac{u f^{r}(\pi_{rr}-k_{rr})\sqrt{1+r^{2}}+u f^{\alpha}(\pi_{r\alpha}-k_{r\alpha})\sqrt{1+r^{2}}}{\sqrt{1 + u^{2}|\nabla f|_{g}^{2}}}
+o(r^{-1})\\
=& \sqrt{1+r^{2}}\left(r^{2}+r(u_{0}+\mathcal{A})+O(r^{\varepsilon}) \right) \left( -\frac{2u_{0}+\mathcal{A}}{r^{3}}+O(r^{-4+\varepsilon})\right)\\
&\cdot\left( \frac{1}{r} - \frac{u_{0}+\mathcal{A}}{r^{2}}+ O(r^{-3+\varepsilon}) \right) +o(r^{-1}) \\
=& -\frac{2u_{0}+\mathcal{A}}{r} + o(r^{-1}),
\end{split}
\end{equation}
and
\begin{equation} \label{011}
q(\nabla_{S} f)= \sum_{\alpha, \beta = \theta, \phi} \frac{u g^{\alpha \beta}f_{\alpha}f^{j}(\pi_{\beta j}-k_{\beta j})}{\sqrt{1 + u^{2}|\nabla f|_{g}^{2}}}
= o(r^{-2}).
\end{equation}
The desired result now follows, and \eqref{010101} vanishes when $m=\overline{m}_{adm}$ since $\mathcal{A}=2m$ and by Lemma \ref{lemma1}, $\overline{m}_{adm}=2m+u_{0}$.
\end{proof}

%We will now show that the boundary integral in \eqref{56} vanishes.
\begin{lemma}\label{lemma4}
Under the hypotheses and notation of Theorem \ref{thm7}
\begin{equation}
 \lim_{r\rightarrow\infty}\int_{\overline{S}_{r}} u \overline{g}(Q,\nu_{\overline{g}}) = 0.
\end{equation}
This relies on the fact that $m=\overline{m}_{adm}$.
\end{lemma}

\begin{proof}
Recall that
\begin{align} \label{401}
\begin{split}
Q (\cdot)
&= (\operatorname{Hess}_{\overline{g}} f)(\by, \cdot)-  \bk(u \barna f, \cdot)
+ (k-\pi)(w, \cdot)+  (k-\pi)(w,w)\frac{u  df}{\sqrt{\volbarg}}, \\
\pi_{ij} &= \frac{u \nabla_{ij}f + u_{i}f_{j} + u_{j}f_{i} + \frac{1}{2}(g_{i \phi}Y^{\phi}_{,j}+g_{j \phi}Y^{\phi}_{,i})}{\sqrt{1+ u^{2}|\nabla f|_{g}^{2}}}, \\
w &= \frac{u \overline{\nabla}f + u^{-1}\byp }{\sqrt{1-u^{2}|\overline{\nabla} f|_{\bg}^{2}}}
      = \frac{u \nabla f + u^{-1}\byp }{\sqrt{1+u^{2}|\nabla f|_{g}^{2}}}.
\end{split}
\end{align}
Therefore \eqref{01} implies
\begin{align} \label{003}
\begin{split}
&\int_{\overline{S}_{\infty}} u \overline{g}(Q,\nu_{\overline{g}})  \\
=&\int_{\overline{S}_{\infty}} u(\operatorname{Hess}_{\overline{g}}  f)(\by, \nu_{\overline{g}})-  \bk(u^{2} \barna f, \nu_{\overline{g}})  \\
&+ \int_{S_{\infty}}  \frac{u(1 + u^{2}|\nabla_{S}f|^{2})}{\sqrt{1+u^{2}|\nabla f|_{g}^{2}}}
\left(\widetilde{Q}(e_{r}) - \frac{u^{2}e_{r}(f)}{1+ u^{2}|\nabla_{S}f|^{2}}\widetilde{Q}(\nabla_{S}f) - \frac{e_{r}(f)}{1+ u^{2}|\nabla_{S}f|^{2}} \widetilde{Q}(Y) \right)
\end{split}
\end{align}
where
\begin{equation}
\widetilde{Q}=(k-\pi)(w, \cdot)+ \sqrt{1+u^{2}|\nabla f|_{g}^{2}} (k-\pi)(w,w) u df.
\end{equation}
Since $\nu_{\overline{g}}$ is the unit normal for an axisymmetric surface we have that
$\bk(u^{2} \barna f, \nu_{\overline{g}}) =0$, and since $\byp = O(r^{-3})$ we have that
$(\operatorname{Hess}_{\overline{g}}  f)(\by, \nu_{\overline{g}})=O(r^{-3})$. It follows that the first integral on the right-hand side of \eqref{003} vanishes.
Furthermore, the integrand of the second integral simplifies to
\begin{align}
\begin{split}
&\frac{(1 + u^{2}|\nabla_{S}f|^{2})}{\sqrt{1+u^{2}|\nabla f|_{g}^{2}}} \left(\widetilde{Q}(e_{r}) - \frac{u^{2}e_{r}(f)}{1+ u^{2}|\nabla_{S}f|^{2}}\widetilde{Q}(\nabla_{S}f) - \frac{e_{r}(f)}{1+ u^{2}|\nabla_{S}f|^{2}} \widetilde{Q}(Y)\right) \\
=& \frac{(1 + u^{2}|\nabla_{S}f|^{2})}{\sqrt{1+u^{2}|\nabla f|_{g}^{2}}} (k-\pi)(w, e_{r})- (k-\pi)\left(w, \frac{u^{2}e_{r}(f)\nabla_{S}f+e_{r}(f)Y}{\sqrt{1+u^{2}|\nabla f|_{g}^{2}}} \right)
+ u(k-\pi)(w,w) e_{r}(f) \\
=&\frac{(1 + u^{2}|\nabla_{S}f|^{2})}{\sqrt{1+u^{2}|\nabla f|_{g}^{2}}} (k-\pi)(w, e_{r}) + (k-\pi)\left(w, \frac{u^{2}e_{r}(f)^{2}e_{r}}{\sqrt{1+u^{2}|\nabla f|_{g}^{2}}}\right) \\
=& \sqrt{1+u^{2}|\nabla f|_{g}^{2}}(k-\pi)(w,e_{r}).
\end{split}
\end{align}
Hence
\begin{align} \label{004}
\int_{\overline{S}_{\infty}} u \overline{g}(Q,\nu_{\overline{g}})  = \int_{S_{\infty}} u\sqrt{1+u^{2}|\nabla f|_{g}^{2}}(k-\pi)(w,e_{r}).
\end{align}

We now compute asymptotics of the tensor $\pi$. Using \eqref{43}, \eqref{52}, \eqref{013}, and \eqref{014} produces
\begin{align} \label{008}
\begin{split}
\pi_{r \alpha} &= \left(\nabla_{r\alpha}f + \frac{\partial_{r}u\partial_{\alpha}f}{u} + \frac{\partial_{\alpha}u\partial_{r}f}{u} + \frac{1}{2u}(g_{\alpha \phi}\partial_{r}Y^{\phi}+g_{r \phi}\partial_{\alpha}Y^{\phi})\right)\left(u^{-2}+|\nabla f|_{g}^{2}\right)^{-1/2}  \\
&= \left(\partial_{\alpha}u- \frac{1}{2}g^{\beta \gamma}\partial_{\beta}f\partial_{r}g_{\gamma \alpha}+O(r^{-2}) \right)\left( \frac{1}{r} + O(r^{-2})) \right) \\
&= O(r^{-2}),
\end{split}
\end{align}
and
\begin{align} \label{009}
\begin{split}
\pi_{\alpha \beta}&= \left(\nabla_{\alpha\beta}f + \frac{\partial_{\beta}u\partial_{\alpha}f}{u} + \frac{\partial_{\alpha}u\partial_{\beta}f}{u}+\frac{1}{2u}(g_{\alpha \phi}\partial_{\beta}Y^{\phi}+g_{\beta \phi}\partial_{\alpha}Y^{\phi})\right)\left(u^{-2}+|\nabla f|_{g}^{2}\right)^{-1/2}\\
&= \left( \frac{1}{2}g^{rr}\partial_{r}f\partial_{r}g_{\alpha \beta} + O(1) \right) \left( \frac{1}{r} + O(r^{-2}) \right) \\
&= r^{2}\sigma_{\alpha \beta} + O(r^{-1}).
\end{split}
\end{align}
Thus with the help of
\begin{equation}\label{000}
e_{r}= \left(\sqrt{1+r^{2}} +O(r^{-2}) \right) \partial_{r} + O(r^{-5}) \partial_{\theta}
+ O(r^{-5}) \partial_{\phi},
\end{equation}
\begin{equation}\label{0000}
e_{\theta}= \left(\frac{1}{r}+O(r^{-4})\right)\partial_{\theta} + O(r^{-5})  \partial_{\phi},  \text{ }\text{ }\text{ }\text{ }\text{ }\text{ }\text{ }
e_{\phi} = \left(\frac{1}{r\sin{\theta}} +O(r^{-4}) \right) \partial_{\phi},
\end{equation}
it follows that
\begin{align} \label{0010}
\begin{split}
\int_{\overline{S}_{\infty}} u \overline{g}(Q,\nu_{\overline{g}})
&= \int_{S_{\infty}} u\sqrt{1+u^{2}|\nabla f|_{g}^{2}}(k-\pi)(w(e_{r})e_{r},e_{r})  \\
&= \int_{S_{\infty}} u^{2}e_{r}(f)(r^{2}+O(1)) (k-\pi)_{rr} .
\end{split}
\end{align}

It is straightforward to check that \eqref{0010} vanishes if $(k-\pi)_{rr}=o(r^{-5})$. To see that this is the case, observe that
\begin{equation} \label{007-1}
\nabla_{rr}f = \partial^2_r f - \Gamma_{rr}^{r}\partial_{r}f + O(r^{-5})
= \frac{1}{r} + \frac{2\mathcal{D}_{1}-1}{2r^{3}} + \frac{3 \mathbf{m}^{r}-2\mathcal{A}+8\mathcal{D}_{2}}{2r^{4}} + O(r^{-5}),
\end{equation}
\begin{align} \label{007-2}
\begin{split}
\frac{2\partial_{r}u\partial_{r}f}{u}
= & -\frac{2\mathcal{C}_{1}}{r^{2}}
  - \frac{4\mathcal{C}_{2}+2\mathcal{A}\mathcal{C}_{1}-2\mathcal{C}_{1}^{2}}  {r^{3}} \\
  &+ \frac{\mathcal{C}_{1}- 6 \mathcal{C}_{3} -4\mathcal{A} \mathcal{C}_{2}+6 \mathcal{C}_{1} \mathcal{C}_{2}+ 2\mathcal{C}_{1}\mathcal{D}_{1}
  +2 \mathcal{A} \mathcal{C}_{1}^{2} - 2 \mathcal{C}_{1}^{3}}{r^{4}} + O(r^{-5}),
%&= \frac{2 \overline{m}}{r^{2}} + \frac{6 \overline{m}^{2}-4C_{2}}{r^{3}} + \frac{-\overline{m}-18\overline{m}C_{2}-6C_{3}+12\overline{m}^{3}}{r^{4}} + O(r^{-5}),%
\end{split}
\end{align}
and
\begin{align} \label{007-3}
\begin{split}
\left(u^{-2}+|\nabla_{g}f|^{2} \right)^{-1/2}
=& \frac{1}{r}
 - \frac{\mathcal{A}}{r^{2}}+ \frac{2\mathcal{D}_{1}+2\mathcal{A}^{2}-1}{2r^{3}}  \\
 &+ \frac{\tfrac{1}{2}\mathbf{m}^{r} + \mathcal{A}+\mathcal{C}_{1}-2 \mathcal{A} \mathcal{D}_{1} + 2 \mathcal{D}_{2}-\mathcal{A}^{3}}{r^{4}} + O(r^{-5}).
\end{split}
\end{align}
Definitions of the coefficients in \eqref{014-1} then imply that
\begin{align}\label{007}
\begin{split}
\pi_{rr}
&= \left(\nabla_{rr}f + 2\frac{\partial_{r}u\partial_{r}f}{u}+g_{r \phi}\partial_{r}Y^{\phi}\right)
\left(u^{-2}+|\nabla_{g}f|^{2} \right)^{-1/2}\\
%&= \frac{1}{r^{2}} - \frac{A+2C_{1}}{r^{3}} + \frac{2D_{1}+A^{2}-4C_{2}+2C_{1}%^{2}-1}{r^{4}} \\
%&+ \frac{6D_{2}-3AD_{1}+A/2-A^{3}+3C_{1}-6C_{3}+6C_{1}C_{2}-2C_{1}^{3}}{r^{5}} %+ O(r^{-6}) \\
&= \frac{1}{r^{2}} -\frac{1}{r^{4}} + \frac{b^{(5)}_{rr}}{r^{5}} +O(r^{-6}) ,
\end{split}
\end{align}
Moreover
\begin{equation}
k_{rr} = \frac{1}{1+r^{2}}  + \frac{b^{(5)}_{rr}}{r^{5}} + o(r^{-5}),
\end{equation}
which leads to the desired conclusion.
%\begin{equation}
%\begin{split}
%&\int_{S_{\infty}} u \sqrt{1+u^{2}|\nabla f|_{g}^{2}}(k-\pi)(w(e_{r})e_{r}, %e_{r}) dS  \\
%&= \int_{S_{\infty}} u^{2}(f_{,r})(1+r^{2})^{3/2}(k-\pi)_{rr} dS \\
%&= 0,
%\end{split}
%\end{equation}
%if and only if $m=\overline{m}$.
\end{proof}

\end{document}